\documentclass{article}

\usepackage{watkins}
\title{Proper Learning of Shallow All-to-All Quantum Circuits}
\author[1,2]{Steven Kordonowy\thanks{skordono@ucsc.edu. Contributions were made entirely during the
author's internship at JPMorganChase.}}
\author[1]{Jacob Watkins\thanks{jacob.a.watkins@jpmchase.com}}
\affil[1]{Global Technology Applied Research, 
JPMorganChase, New York, NY 10001, USA}
\affil[2]{University of California, Santa Cruz, CA 95060, USA}

\date{\today}

\usepackage{xpatch}
\xapptocmd{\appendices}{%
    \crefalias{section}{appendix}%
    \crefalias{subsection}{appendix}%
}{}{\PatchFailed}

\usepackage[figure,boxed,linesnumbered]{algorithm2e}
\DontPrintSemicolon
\SetKwInOut{Input}{Input}
\SetKwInOut{Output}{Output}
\SetKwComment{Comment}{/* }{ */}
\newcommand{\NearComment}[1]{%
    \hspace{1em}\Comment*[f]{#1}%
}
\SetKw{Parameters}{Parameters : }
\SetKwFunction{FactorizeFrontGate}{FactorizeFrontGate}
\SetKwFunction{ForwardLearn}{ForwardLearn}
\IncMargin{1em}
\let\oldnl\nl
\newcommand{\nlnonumber}{\renewcommand{\nl}{\let\nl\oldnl}}

\newcommand{\backvec}[1]{\reflectbox{$\vec{\reflectbox{$#1$}}$}}

\usepackage{quantikz}

\begin{document}

\maketitle

\begin{abstract}
    We study proper learning of shallow quantum circuits given query access and a known gate layout,
    seeking an operationally equivalent circuit with the same layout and gates from a prescribed
    family. Building on the local-inversion method of Fefferman, Ghosh, and Zhan~\cite{fefferman2024anti}, we identify a
    meta-algorithmic framework for learning broad classes of circuits based on iterative local gate
    inversions at the front and back of the circuit. For forward learning on random all-to-all
    $k$-local layered architectures with fixed $k\ge2$ and no repeated
    qubit pairs in consecutive layers, we prove a transition in lightcone structure within a
    constant-width window around $\log_k n+\log_k\ln n$. With high probability, causal boundaries
    remain available for exact recursive factorization below this window, whereas above it no outer
    gate admits the causal witness needed to initiate local inversion. For two-qubit gates drawn
    independently from a suitable well-spaced Haar discretization, we obtain exact proper learning
    with high probability using polynomially many queries below this window, given the layout,
    finite gate list, and exact implementations of the listed inverses. Numerical experiments explore
    independent two-qubit layers and alternating forward--backward removal.
\end{abstract}

\section{Introduction} \label{sec:introduction}

Inferring the interactions that generate an observed quantum process is a natural task in quantum
information~\cite{bairey2019learning,bakshi2024structure,gu2024practical}. For a process generated by
a shallow quantum circuit, the underlying description is especially concise: a small number of gates, 
each acting on only a few qubits. Yet these local interactions can produce
entanglement and distribute information across the system, so a short description does not, by
itself, provide an efficient route to reconstruction. Understanding when this description can be
learned connects the study of quantum dynamics with questions about the computational difficulty
of inference.

Huang et al.~\cite{huang2024learning} establish that shallow quantum circuits can indeed be learned
efficiently: for an unknown constant-depth circuit on $n$ qubits, they give a polynomial-time classical algorithm
that uses single-qubit measurement data to construct a constant-depth approximation to its unitary
channel in diamond norm, using an additional $n$-qubit register. Related works study the learning of states prepared 
by shallow circuits~\cite{landau2025learning,kim2024learning}. While these learning protocols reproduce the circuit's
operation or output state, they do not necessarily produce a circuit with the
original gate layout and gate family; rather, the reconstruction generally requires more qubits or greater depth
than the original circuit. 

The desire to retain structural features of the original circuit motivates \emph{proper learning}, in which the learner
must output a circuit from a prescribed family; a circuit outside that family is \emph{improper}.
This task is at least as hard as improper learning, and appears difficult to solve given only
query or state access. Therefore, this work studies proper learning given a known gate layout
in addition to the usual queries. We seek an operationally equivalent circuit with the supplied layout and gates from the same
family, without requiring recovery of the original gate labels.

Restrictions on the learned circuit layout arise naturally from several recently proposed
 quantum cryptographic protocols.
These protocols rely on the difficulty of inferring a valid circuit from public output
states~\cite{fefferman2025hardness} or classical shadows~\cite{niroula2026digital}.
Protocols that inspect a submitted circuit description and enforce depth or ancilla restrictions
can reject otherwise accurate improper representations. Our results concern the stronger model of
known-layout query access, however, and do not directly establish attacks on these public-data
protocols. Within this model, we identify the causal structure needed for local-inversion learning
beyond geometrically local architectures.

Our approach builds on the local-inversion method of
\citeauthor{fefferman2024anti}~\cite{fefferman2024anti}, who establish efficient proper learning
for logarithmic-depth one-dimensional brickwork circuits with gates drawn from suitable discrete
random ensembles. Their method uses the known circuit layout to test candidate inverses of
individual gates, progressively reconstructing the circuit while preserving its architecture.
Its success depends on the circuit's causal structure and on signals remaining detectable as
they propagate through the circuit. The authors also describe extensions beyond one dimension
to a ``wide class of architectures with good lightcone properties.'' This raises the structural
question we address here: which features of the circuit layout support local inversion when
geometric locality is absent?

\subsection*{Overview of present work and results}
We develop a local-inversion framework for proper learning of circuits with known layouts.
We identify \emph{causal boundaries} as witnesses for testing local gate factorizations and
characterize their role in removing outer gates. This framework only relies on bounded interaction locality
of the gates, not on any particular geometric arrangement of the qubits.

After introducing the necessary concepts and terminology, we turn to the question of algorithmic
implementation. We argue that, in addition to the ``forward learning'' approaches considered
in~\cite{fefferman2024anti}, a learner may also have access to ``backward'' learning by using the
backward lightcones. For general gate layouts, we provide examples where iterative forward-backward
learning is strictly more powerful than a forward-only approach. We also identify three essential
criteria for a local factorization learning protocol to succeed.

Finally, we describe lightcone evolution in layered circuits by tracking the growth of individual
cones and the symmetric difference between pairs. We apply this description to random all-to-all
$k$-local layered architectures, for every fixed $k\ge2$.
We exclude repeated qubit pairs in consecutive layers, and prove a
transition within a constant-width window around a \emph{threshold}
\begin{equation}
    t(n,k) \coloneqq \log_k n+\log_k\ln n.
\end{equation}
Below this window, the forward lightcones immediately after every gate are pairwise distinct on
its support, providing the causal boundaries needed throughout exact recursive factorization.
Above it, the circuit has neither a front pivot nor a back pivot: removing any gate at the front
or back leaves input-output reachability in the circuit graph unchanged. The leading $\log_k n$ term reflects the
maximum lightcone growth rate allowed by $k$-locality. The additional $\log_k\ln n$ term
reflects the slower approach to full coverage: random gate placement makes no purposeful attempt
to connect qubits across the lightcone boundary and therefore does not maximize growth at each layer. 

These conclusions
establish the structural scope of the local inversion framework; implementing a learning protocol
additionally requires suitable signal propagation and successful handling of the resulting factors.
For two-qubit gates, we combine the lightcone threshold theorem with the signal-propagation estimates of 
\cite{fefferman2024anti}, adapted to a specified Haar discretization,
to obtain exact proper learning with polynomial query complexity below the window
(\Cref{cor:single_overlap_learning}). The guarantee assumes that the gate list and exact inverse
implementations are supplied.

\subsection*{Connections to information spreading}
Our lightcone analysis connects to classical models of information spreading.
At the level of circuit architecture, each gate pools the sets of input labels that can reach
its participating wires, just as a group gossip exchange pools its participants' knowledge.
Such exchanges are described by \emph{gossip monoids} through compositions of equivalence
relations~\cite{brouwer2015lossy,fenner2018gossip}. 
This connection also yields a deterministic size bound: an $n$-qubit circuit consisting of
two-qubit gates that admits complete forward-only pivot removal has at most $\binom{n}{2}$
gates (\Cref{cor:forward_gate_bound}). For layered circuits, this implies $d\le n-1$,
giving a universal depth bound that complements the logarithmic threshold for our random
ensemble.

Randomized rumor spreading is a single-message gossip problem: one tracks which participants
know a particular message, rather than the full collection of secrets held by each
participant~\cite{karp2000randomized,doerr2017randomized}. In the pull protocol, each uninformed
participant contacts a random participant and requests the message. Doerr and
Kostrygin~\cite[Theorem~57]{doerr2023randomized} prove that its expected completion time on the complete
graph is $\log_2 n+\log_2\ln n+O(1)$, with the same logarithmic and double-logarithmic terms as
our lightcone threshold for $k=2$. Exponential growth followed by doubly exponential decay of
the uninformed population explains this shared scale. Our theorem concerns the persistence of
distinguishing outputs in a different random contact process: below the threshold window,
every pair of wires on every gate has different forward lightcones, simultaneously with high
probability. Single-rumor completion estimates do not imply this property, since two lightcones
may coincide before either reaches full coverage. Above the window, we establish simultaneous
saturation sufficient to eliminate all initial front and back pivots.

Conversely, the lightcone threshold theorem (\Cref{thm:lightcone_threshold}) has a gossip
interpretation. Assign a distinct message to
each participant and let every scheduled group pool its members' knowledge. For a schedule
drawn from our single-overlap ensemble below the threshold window, the participants in every
exchange have pairwise-distinct knowledge sets immediately before that exchange, simultaneously
with high probability. In this sense, every meeting is informative: at least one participant
learns a message they did not previously know. More generally, every nonsingleton subgroup of
a meeting contains a participant who would learn something by pooling knowledge within that
subgroup.

In graph-theoretic terms, full lightcones from every qubit immediately after a given layer
correspond to temporal
connectivity: every qubit can reach every other through a time-respecting sequence of gates.
Thresholds for this property have been studied in other random
contact models~\cite{casteigts2024sharp}.

\subsection*{Paper organization}
Following the preliminaries in~\Cref{sec:preliminaries}, \Cref{sec:causality} introduces lightcones,
causal boundaries, and pivot gates. \Cref{sec:algorithm} presents the local-inversion learning
framework, its conditions for success, and its limitations. \Cref{sec:random_circuit_ensemble}
analyzes lightcone evolution in layered circuits, establishes a threshold for random all-to-all
architectures, and states the associated proper-learning guarantee for discretized Haar gates.
It also reports numerical experiments beyond the ensemble covered by the theorem.
\Cref{sec:discussion} concludes with a discussion.

\section{Preliminaries} \label{sec:preliminaries}

    We write $[n]$ for the first $n$ positive integers. Let $\cal{U}(N)$ denote the group of unitaries on
    $\bb{C}^N$. For eventually positive $g$, we write $f=O(g)$ if $|f(n)|\le Cg(n)$ for some
    constant $C>0$ and all sufficiently large $n$, and $f=o(g)$ if $f(n)/g(n)\to0$ as
    $n\to\infty$. Subscripts indicate permitted dependence of the implicit constants on fixed
    parameters. An event $E_n$ occurs \emph{with high probability} if
    $\Pr(E_n)=1-o(1)$ as $n\to\infty$.
    For falling factorials, write $(x)_r=x(x-1)\cdots(x-r+1)$ for integers $r\ge1$,
    with $(x)_0=1$.

    The \emph{architecture} (or \emph{layout}) of a (unitary) quantum circuit $C$ is represented by
    a directed acyclic graph whose vertices represent gate locations and whose edges represent
    qubit wire segments. Assigning a unitary operation from a prescribed family $\cal{G}$ to each
    gate vertex specifies the circuit. We use $C$ for both the circuit and its architecture when
    the distinction is clear. Because we are interested in invertible circuits and
    gates, the sets of input and output edges through each gate are the same. Thus, the working
    computational space remains unchanged throughout, and we can assign a label $q_1, q_2,
    \ldots, q_n$ to each edge in the circuit. This collection $Q = \{q_i\}_{i\in[n]}$ denotes
    the set of qubits. In principle, one could also work with qudits of internal dimension $d>
    2$ and much of the analysis will be the same. The input and output edges for $C$ are
    ``hanging,'' or may be considered connected to special input and output vertices.
    \Cref{fig:circuit_graph_schematic} illustrates the basic correspondence between circuit diagrams 
    and their graph representations.
    \begin{figure}
        \centering
        \includegraphics[width=\linewidth]{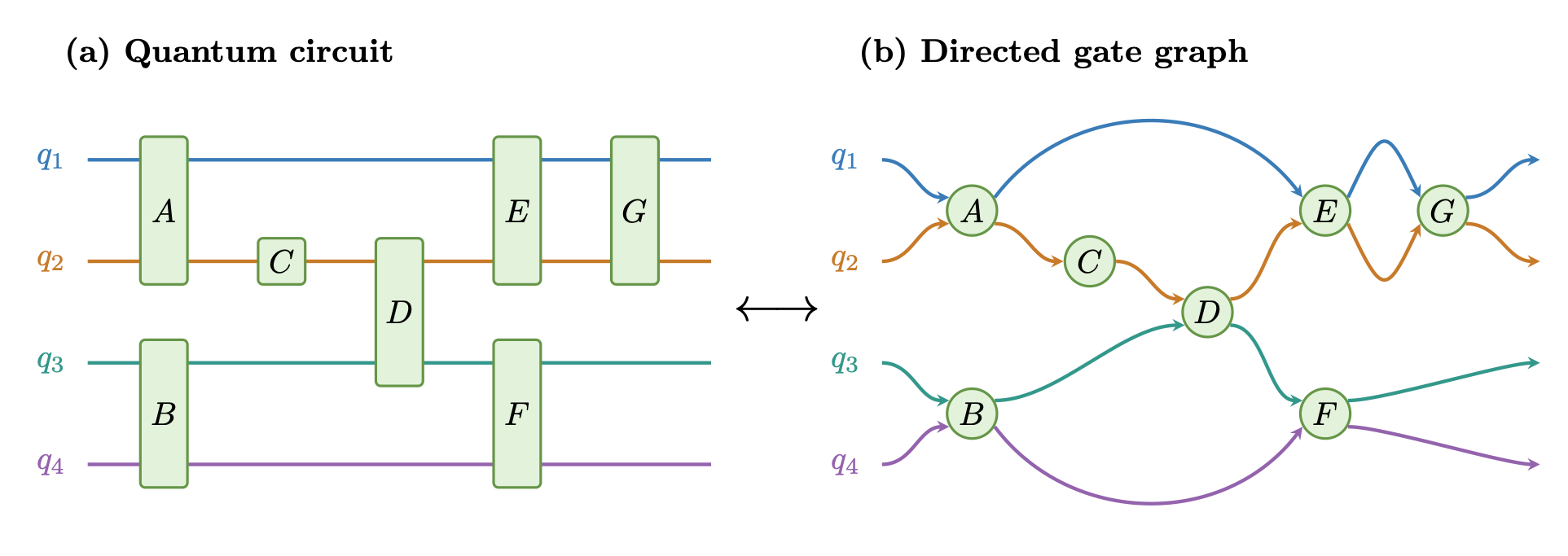}
        \caption{The same four-qubit circuit in (a) standard circuit notation and (b) directed
        graph form. Gates become labeled vertices and wire segments become directed edges;
        colors identify the qubits in both representations.}
        \label{fig:circuit_graph_schematic}
    \end{figure}

    Every circuit $C$ induces a unitary map $U_C$ on the state space $\bb{C}^{2^n}$
    of the qubits. Fixing a state on the remaining inputs and tracing out unwanted outputs
    defines an induced quantum channel between subsystems. We write $\cal{E}_C$ when the input
    and output subsystems and the fixed input state are clear from context.
    For a gate $G$, let $\partial G\subseteq Q$ denote the qubits incident on $G$.
    By \emph{removal} of a gate $G$ from $C$, denoted $C\setminus G$, we mean in the usual sense
    of quantum circuit diagrams, with $G$ replaced by wires for each qubit. Graph-theoretically,
    this means joining each incoming edge to the outgoing edge carrying the same qubit label, rather
    than simple vertex deletion.
    
    The distance between two edges $\dist(e, e')$ is the minimum number of gates traversed by
    a directed path from $e$ to $e'$ in $C$, or $+\infty$ if no such path exists. A path from
    $q\in Q$ to $q'\in Q$ means a directed path from the input wire of $q$ to the output wire
    of $q'$. Overloading notation, $\dist(q,q')$ is the distance between these input and
    output wires. It is not symmetric in general due to directedness of $C$. The (maximum)
    \emph{depth} of a circuit $C$ is the maximum number of gates traversed by any directed
    path through $C$.
    
    For the lightcone-growth analysis in~\Cref{sec:random_circuit_ensemble}, we specialize to \emph{layered}
    circuit, i.e., ones whose gates can be partitioned into ordered layers such that the gate supports in
    each layer partition $Q$. Although inserting identity
    gates can make any circuit layered, we work with the original architecture and do not insert
    gates on idle wires.

\section{Qubit lightcones and causal influence} \label{sec:causality}
    From one perspective, the defining property of shallow circuits is that influence between
    qubits is restricted: a local patch of qubits does not have influence over, and is not
    influenced by, every other qubit. This idea is central to most existing shallow-circuit
    learning protocols, including ours. Thus, we take this section to carefully define some
    concepts related to causality and influence, with respect to graph connectivity.

    Path connectedness to and from a given qubit in a circuit graph $C$ is captured by the notion of
    \emph{lightcones}, borrowing physics terminology. We write $q(0)$ for the input wire of $q$
    and $q(\ell)$ for its wire segment immediately after the $\ell$th gate acting on it.
    In a layered circuit, $\ell$ is also the layer index.
    \begin{definition}[Lightcones]\label{def:LC_CS}
        For a wire $q(\ell)$, the forward (backward) lightcone $\vec{L}_{\ell,d}(q)$
        ($\backvec{L}_{\ell,d}(q)$) is the set of qubits in the smallest subgraph of $C$ containing
        all paths from (to) wire $q(\ell)$ with distance at most $d$. When $d$ is the full circuit
        depth, we drop $d$ in the above expressions.
    \end{definition}
    \noindent For some readers, this definition may appear unusual, as it only considers the
    influenced qubit set and not a subgraph. However, because this work only requires the
    influenced set of qubits, and not the subgraph structure, we opt for familiar and mnemonic
    terminology. Observe that these lightcones are naturally nested by depth, and can be built up
    recursively and efficiently provided the circuit is polynomially sized. While the concept of
    a lightcone is frequently invoked in settings of geometric locality, due to restrictions on
    causal influence, we remark that a similar restriction may arise merely from imposing
    interaction locality (in our setting, restricting the degree of each vertex gate).

    In order to learn $C$, we probe the causal structure in the circuit as exhibited by the
    collection of qubit lightcones. The boundary of a lightcone is the set of newly influenced qubits following
    the application of a gate. This provides the key mechanism for the learning protocol.
    \begin{definition} \label{def:causal_boundary}
        The forward \textbf{causal boundary} of $q$ at $\ell \in \bb{Z}$ is defined as the set
        difference of forward lightcones with roots immediately before and after the same gate.
        \begin{equation}
            \partial \vec{L}_\ell(q) \coloneqq \vec{L}_{\ell-1}(q)\setminus \vec{L}_\ell(q)
        \end{equation}
        Similarly, the backward causal boundary is defined by $\partial\backvec{L}_\ell(q)
        \coloneqq \backvec{L}_\ell(q) \setminus \backvec{L}_{\ell-1}(q)$.
    \end{definition}
    The above boundaries move the lightcone's root across a gate. Alternatively, one can fix the
    root and grow the lightcone across a causal cut. \Cref{app:causal_cuts} develops this description
    for arbitrary circuit architectures and shows that forward growth and backward changing-root
    boundaries identify the same qubit pairs, and conversely. In layered circuits, cuts between
    layers recover the depth-changing formulation. The forward and backward boundaries
    in~\Cref{def:causal_boundary} still provide
    independent and useful information, as we will see below.
    
    When the causal boundary is not empty, this implies the gate $G$ crossed along qubit $q$
    strictly increases lightcone size. When $G$ is at the front or back of the circuit, removing
    it disconnects an input-output pair. This disconnection provides the basis for testing local
    inversions. We provide some terminology for such gates.
    \begin{definition} \label{def:pivot_gate}
        An outer gate in $C$ is simply a gate at the front or back of the circuit. An outer gate $G$
        that has some qubit $q \in \partial G$ with nonempty causal boundary is called a
        \textbf{pivot gate}.
    \end{definition}
    \noindent This definition applies symmetrically to forward or backward lightcones.
    The following lemma makes the gate-removal interpretation precise: removing the gate disconnects
    exactly the pairs identified by its causal boundary. In the backward statement, the index
    $-1$ indicates that the causal boundary is taken across the last gate acting on output
    qubit $q'$.
    \begin{lemma}[Lightcones and influence testing] \label{lem:causal_boundary_iff_gate_remove}
        Let $G_f$ be the first gate acting on input qubit $q$, and let $G_b$ be the last gate
        acting on output qubit $q'$. Then $q'\in \partial \vec{L}_1(q)$ iff $q$ is connected to
        $q'$ in $C$, but removing $G_f$ disconnects $q$ and $q'$. Similarly,
        $q \in \partial \backvec{L}_{-1}(q')$ iff $q$ connects to $q'$ but not when $G_b$ is
        removed.
    \end{lemma}
    \begin{proof}
        We prove the forward statement; the backward statement follows by reversing the direction
        of the circuit graph. Removing $G_f$ replaces it by wires joining each input
        to the corresponding output. Thus, in $C\setminus G_f$, input $q$ continues directly to
        $q(1)$, and its reachable output qubits are exactly those in $\vec{L}_1(q)$ of the
        original circuit. Before removal, its reachable outputs are $\vec{L}_0(q)$. Consequently,
        the outputs disconnected from $q$ by removing $G_f$ are precisely
        \begin{equation}
            \vec{L}_0(q)\setminus\vec{L}_1(q)=\partial\vec{L}_1(q).
        \end{equation}
        This proves the equivalence.
    \end{proof}
    \noindent The lemma identifies an input-output pair whose connection disappears when a pivot
    gate is removed. We call such an input-output pair a \emph{witness} for the removed gate.
    When the root qubit is understood, we also call the qubit in its causal boundary a witness.
    Since disconnection rules out physical influence, this provides a structural
    basis for testing local inversions, illustrated in~\Cref{fig:learning_schematic}.
    \Cref{sec:algorithm} develops the learning procedure and the conditions under which such tests
    are reliable.
    \begin{figure}
        \centering
        \includegraphics[width=\textwidth]{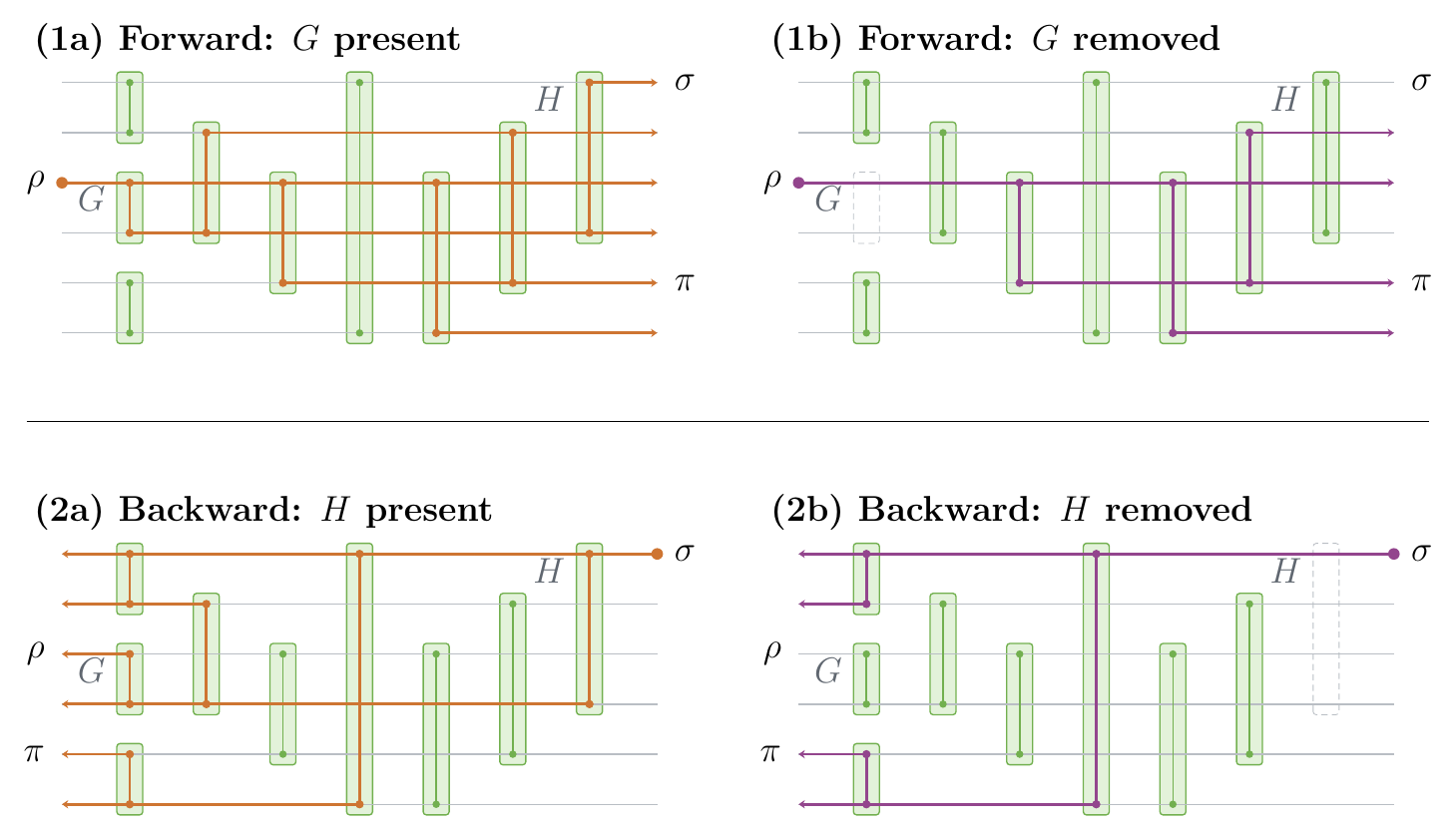}
        \caption{Lightcones on the same six-qubit circuit, with a pivot gate present (orange)
        or removed (purple). Top row: the forward lightcone from input $\rho$, with front gate
        $G$ present (1a) or removed (1b). Bottom row: the backward lightcone from output $\sigma$,
        with back gate $H$ present (2a) or removed (2b). In both rows, removing the pivot
        disconnects $\rho$ from $\sigma$, making $(\rho,\sigma)$ a witness for testing local
        factorization. By contrast, $\pi$ remains in the relevant lightcone after removal
        and does not provide a witness.}
        \label{fig:learning_schematic}
    \end{figure}
    
    Geometric locality can make connections at the edge of lightcone growth witnesses for both
    front and back pivots. On a lattice, for example, gates of bounded spatial range limit how
    far a lightcone can spread at each step. A connection near the limit of this spread can
    depend on both its first and last gates: removing either leaves too few steps to bridge
    the separation. The following proposition gives a sufficient condition for this to occur.
    \begin{proposition}[Pivots at the edge of spatial reachability]
    \label{prop:front_back_pivot}
        Let $C$ be a circuit of maximum depth $d\ge1$, and equip its qubit set $Q$ with a metric
        $\rho$. Suppose that, for some $R>0$, every gate $G$ satisfies
        \begin{equation}
            \max_{a,b\in\partial G}\rho(a,b)\le R.
        \end{equation}
        If input $q$ reaches output $q'$ and $\rho(q,q')>(d-1)R$, then the first gate $G_f$ acting
        on $q$ and the last gate $G_b$ acting on $q'$ are front and back pivots, respectively,
        both witnessed by $(q,q')$.
    \end{proposition}
    \begin{proof}
        Consider any directed path from input $q$ to output $q'$ traversing $m$ gates. Let
        $r_0=q,r_1,\ldots,r_m=q'$ be its successive wire labels. Each gate changes the wire
        label by metric distance at most $R$, so the triangle inequality gives
        \begin{equation}
            (d-1)R<\rho(q,q')\le\sum_{i=1}^m\rho(r_{i-1},r_i)\le mR\le dR.
        \end{equation}
        Thus every such path traverses exactly $d$ gates. Its first and last gates are $G_f$
        and $G_b$. If $G_f$ had a gate predecessor on any of its input wires, an input path
        through that predecessor followed by the chosen path from $G_f$ onward would traverse
        more than $d$ gates. Hence $G_f$ is a front gate. Reversing this argument shows that
        $G_b$ is a back gate.

        Suppose a path from $q$ to $q'$ survived removal of $G_f$. Restoring $G_f$ as a
        straight-through passage on wire $q$ would give a path in $C$ with one additional gate.
        The surviving path therefore traverses at most $d-1$ gates, and its endpoint distance
        is at most $R(d-1)$, contradicting $\rho(q,q')>(d-1)R$. The same argument applies to removal
        of $G_b$, restoring a straight-through passage on wire $q'$.
        By~\Cref{lem:causal_boundary_iff_gate_remove}, $(q,q')$ witnesses both pivots.
    \end{proof}
    \noindent The support-diameter bound expresses geometric locality: each gate acts on sites
    at mutual distance at most $R$. For nearest-neighbor two-qubit gates on a lattice, take
    $\rho$ to be lattice graph distance and $R=1$. In a depth-$d$ circuit, any reachable
    input-output pair at lattice distance $d$ then witnesses both a front pivot and a back pivot.
    However, the proposition applies to any metric on the qubit set and requires no Euclidean embedding.

\section{Learning via local inversion} \label{sec:algorithm}

This section develops a local-inversion framework for circuits with known gate layouts. We
describe the factorization procedure and the requirements for iterating it, then prove a
forward-learning guarantee for a well-spaced two-qubit gate family. We also discuss extensions
beyond this setting.

One query is one application of the unknown unitary channel
$\rho \mapsto U_C\rho U_C^\dagger$ to the circuit register, acting as the identity
on any ancillary registers. The learner may prepare arbitrary mixed or entangled states and
perform known quantum operations, measurements, and classical processing between queries.
Access to $U_C^\dagger$ or controlled-$U_C$ is not assumed. Unless stated otherwise, our
learning complexity bounds count oracle applications, not classical computation or the cost
of state preparation, measurements, and known operations. Specific guarantees below state
the required implementations of known gates.

We measure output error using diamond norm distance:
\begin{equation}
    d_\diamond(\cal{E},\cal{F})\coloneqq
    \max_{\rho}\norm{(\cal{E}\otimes\mathrm{id})(\rho)
    -(\cal{F}\otimes\mathrm{id})(\rho)}_1.
\end{equation}
Here $\rho$ ranges over density operators on the input register and a reference register
of the same dimension, $\mathrm{id}$ is the identity channel on the reference register,
and $\norm{\cdot}_1$ is the trace norm. Gates and circuits are identified up to global phase;
exact learning means equality of the induced channels. 

\begin{problem}[Circuit Learning given gate layout] \label{prob:black_box_learning}
    Given query access to a circuit $C$, a classical description of its circuit graph and
    gate family $\cal{G}$, a target accuracy $\varepsilon_{\mathrm{out}}>0$, and a failure
    probability $\eta\in(0,1)$, output a classical description of a circuit $C'$ with the
    same circuit graph and gates from $\cal{G}$ such that, with probability at least $1-\eta$,
    \begin{equation}
        d_\diamond(U_{C'},U_C)\le\varepsilon_{\mathrm{out}}.
    \end{equation}
\end{problem}
\noindent We use $U_{C'}\approx U_C$ as shorthand for this channel-error criterion.
The known-layout assumption holds directly for prescribed architectures such as
one-dimensional brickwork. In other structured families, the layout may be selected from a small
set of possibilities. When the architecture is unknown, identifying gate supports is an
additional inference task.

Our framework for solving~\Cref{prob:black_box_learning} builds on the local-inversion methods of
\citeauthor{fefferman2024anti}~\cite{fefferman2024anti}. A related strategy is outlined in
Supplement II of~\cite{niroula2026digital}. Given the
gate layout of $C$, the learner's strategy is to identify the location of a pivot gate $G$, prepend
(or append, for $G$ in back) to $C$ some trial factorizing gate $G'^\dagger$, and test if $q'$ is no
longer causally influenced by $q$. That is, one seeks to induce an approximate factorization $G
G'^\dagger \approx U_1 \otimes U_2$ for $U_1, U_2$ acting along some bipartition of $\partial G$,
which disconnects $q$ from $q'$. This event is detected through single-qubit tomography between an
input qubit $q$ and output $q'$: for distinct choices of initial state $\rho_1, \rho_2$ at input
$q$, perhaps with strategically chosen input states on the remaining qubits, one checks whether the
corresponding output states $\sigma_1, \sigma_2$ are either essentially the same or distinct.
Testing factorization can, depending on the setting, be done near perfectly or with small errors;
we will discuss these possibilities below. \Cref{fig:single_factorization} provides pseudocode for
the protocol described above, termed \texttt{FactorizeFrontGate}. We focus on front gate factorization,
both for concreteness and because of an asymmetry between forward and backward pivot factorization
induced by the tomography protocol.
\begin{figure}
    \centering
    \begin{algorithm}[H]
        \nlnonumber\FactorizeFrontGate \;
        \Input{Black box query access to $U_C$; gate family $\cal{G}$ for $C$; $k$ input/output
        qubits $\partial G = (q, q^c)$ for front pivot $G$, corresponding to input/output pair
        $(q,q')$.}
        \Output{Gate $G' \in \cal{G}$ such that $G G'^\dagger \approx U_1\otimes U_2$.}
        \nlnonumber\Parameters{Factorization precision $\epsilon_f$, tomography precision
        $\epsilon_t$.}
        \BlankLine
        \For{$G' \in \cal{G}$}{
            $U \gets  U_C G'^\dagger$\;
            \For{$(\sigma, \sigma^c) \in \{X,Y,Z\}\times \{I,X,Y,Z\}^{\otimes {k-1}}$}{
                \BlankLine
                initialize $Q \gets \ket{0}^{\otimes n}$\;

                \For{$b \in \{0,1\}^{k-1}$}{
                    \tcp*[l]{Fix other input qubits according to $b$}
                    \For{$j \in [k-1]$} {
                        \eIf{$\sigma_j^c \neq I$}{
                            $q^c_j \gets (I + (-1)^{b_j}\sigma^c_j)/2$}
                            {$q^c_j \gets I/2$}
                    }
                    \tcp*[l]{Toggle $q$ and measure effect}
                    set $q \gets (I+\sigma)/2$\;
                    $\pi^+ \gets$ SingleQubitTomography$(q', U, Q, \epsilon_t)$\;
                    \BlankLine
                    set $q \gets (I-\sigma)/2$  \Comment*[f]{Flip pivot qubit $q$}\;
                    $\pi^- \gets$ SingleQubitTomography$(q', U, Q, \epsilon_t)$\;
                    \BlankLine
                    \If{$\dist(\pi^+, \pi^-) \geq \epsilon_f$}{
                        reject $G'$ and continue to the next $G' \in \cal{G}$\;
                    }
                }
            }
            return $G'$\NearComment{All configurations passed}\;
        }
        null return\tcp*[l]{Protocol fails if this point reached}
    \end{algorithm}
    \caption{Heuristic protocol for reducing a front pivot gate to a product of gates with smaller
    support, via local factorizations. The meta-algorithm iterates over a set of trial
    factorizations $G'^\dagger$ until one of them appears to break the connection between $q$ and
    $q'$ for every tested input configuration. Specifically, the SingleQubitTomography protocol
    tests whether varying $q$ affects $q'$, for each tested fixed configuration on $q^c$.}
    \label{fig:single_factorization}
\end{figure}

A successful local factorization replaces a pivot gate by residual gates acting on disjoint,
smaller supports. To continue reducing $C$, we must handle these residual factors and thereby
convert the factorization into an effective removal of the original gate.

For $k=2$, the residual factors are single-qubit gates, which can be learned separately or absorbed
into subsequent gates when the gate family permits. For $k>2$, one possible implementation
recursively factorizes the residual gates: for each factor acting on two or more qubits,
select a new pivot pair and repeat the local factorization protocol on its support.
The recursion reduces the factors to
single-qubit gates, which are then handled as above. The conditions needed to carry out these
steps are discussed in the following subsection.

In the schematic below, we represent factor handling by the subroutine \texttt{HandleFactors} and
assume that both local factorization and factor handling succeed. The learner then proceeds
through the circuit, selecting pivots and effectively removing gates until no further pivots
remain. If the remaining circuit can be fully inverted, the accumulated corrections define a
candidate circuit $C'$ with $U_{C'} \approx U_C$; otherwise, the procedure gives only a partial
inversion. For concreteness, \Cref{fig:forward_learn} presents the forward-only procedure,
\texttt{ForwardLearn}.
\begin{figure}
    \centering
    \begin{algorithm}[H]
        \nlnonumber\ForwardLearn \;
        \Input{Black box query access to $U_C$. Description of gate layout for $C$.}
        \Output{Description of circuit $C'$ (circuit graph with labeled gates) with same circuit
        graph as $C$ such that $U_C \approx U_{C'}$.}
        \nlnonumber\Parameters{Factorization precision $\epsilon_f$, tomography precision
        $\epsilon_t$.}
        \BlankLine
        $\texttt{Compress}(C)$\tcp*[l]{Group together consecutive gates, if any.}
        Compute all forward causal sets for $C$.\;
        partial\textunderscore only $\gets$ false\;
        \tcp*[l]{Iterate over gates in $C$, front to back.}
        \For{$G$ in $C$(front:back)} {
            \If{$G$ is not a pivot in $C$} {
                partial\textunderscore only $\gets$ true\;
                continue\NearComment{Continue to invert what is possible}\;
            }
            Determine pivot pair $(q,q')$ for $G$, and $\{q,q^c\} = \partial G$\;
            $G' \gets \FactorizeFrontGate(C, \partial G, q, q')$\;
            $C$.prepend($G'^\dagger$) and $C'$.append($G'$)\;
            \texttt{HandleFactors};
        }
        Invert trivial parts of remaining circuit (if applicable, e.g. single-qubit rotations).\;
        \BlankLine
        \If{partial\textunderscore only} {
            print: Warning, only partial inversion possible.\;
        }
        \BlankLine
        \texttt{Decompress}($C'$) \tcp*[l]{Arbitrarily decompose any grouped gates from
        \texttt{Compress}}
        return $C'$\;
    \end{algorithm}
    \caption{Forward-only learning protocol for a unitary circuit $C$, given oracle access and the
    circuit graph. After compressing sequential gates and computing lightcones for each qubit, one
    iterates the \texttt{FactorizeFrontGate} protocol with associated estimation parameters
    $(\epsilon_t, \epsilon_f)$ chosen according to requirements. With each successful
    factorization, the corresponding gate $G'$ is added to the running result circuit $C'$. Final
    decompression consists of choosing an arbitrary sequence for $G'$ to match the layout of
    uncompressed $C$. We assume each \texttt{FactorizeFrontGate} is successful, and
    \texttt{HandleFactors} is the suitable method for converting a factorization into an effective
    removal of the gate.}
\label{fig:forward_learn}
\end{figure}
Observe the inclusion of submodules \texttt{Compress} and \texttt{Decompress}, which are used to
handle cases where two consecutive gates $G_1, G_2$ (or more) are present. Such gates immediately
prevent good lightcone structure, since $G_2$ does not change the lightcone. Indeed, it is not clear
how \emph{any} protocol could distinguish $G_1, G_2$ from some other sequence $G_1', G_2'$ such that
$G_2 G_1 = G_2' G_1'$ within the query framework. Thus, the original individual gate labels
cannot in general be recovered in such cases. We instead attempt to factorize the composite
gate $G_2 G_1$.
\texttt{Compress} merely identifies such compositions as a single gate. When such composite gates
are successfully learned, one then splits each into two or more gates, arbitrarily within the
available set $\cal{G}$, to match the original gate layout. The same idea can be extended to longer
strings of gates.
    
\subsection{Implementation details and conditions for success}
    At a high level, for successful implementation of the above learning schematic, we identify
    three requirements.
    \begin{itemize}
        \item Good lightcone structure: At every stage of the algorithm, there exists a pivot gate. 
        \item Good signal propagation: One must be able to \emph{detect changes} at the output of
        the circuit when local changes are made at the front of the circuit.
        \item Handle factorizations: Once a gate has been successfully factored, one needs to be
        able to proceed to the next iteration successfully.
    \end{itemize}
    The first of these criteria is perhaps the easiest to formalize. 
    \begin{definition}[Good lightcone structure] \label{mydef:good_lightcones}
        A circuit $C$ has \textbf{good lightcone structure} if there exists an ordering $G_1,
        \ldots, G_m$ of all the gates of $C$, such that $G_j$ is a pivot gate for $C$ when all $G_i$
        are removed for $i < j$. We say circuit $C$ has \textbf{good forward (backward) lightcone
        structure} if such an ordering exists, and $G_j$ is always a front (back) pivot gate.
    \end{definition}
    \noindent One-dimensional brickwork circuits with $d<n$ have good lightcone structure
    because their forward (or backward) lightcones
    strictly grow across each layer. More generally, \Cref{prop:front_back_pivot} provides a
    geometric criterion for finding front and back pivots. When this criterion remains
    applicable throughout a complete gate-removal sequence, it establishes good lightcone
    structure. Good lightcone structure implies no consecutive gates act on a
    common set of qubits, but in practice this may frequently occur. We handle this with a
    ``compression'' step.
    
    For forward learning, the relevant causal-boundary condition can be read off from $C$ at the
    outset. The gates removed during a forward pass form a predecessor-closed set, and removing
    them does not change the forward causal boundary of any remaining gate. Hence, a gate $G$ is a
    front pivot when it reaches the front iff it has a nonempty forward causal boundary at its
    original location in $C$. Consequently, $C$ has good forward lightcone structure if and only
    if every gate has a nonempty forward causal boundary in $C$.

    Complete forward pivot removal requires the two output lightcones immediately after every
    gate to differ. Viewed backward from the outputs, each gate therefore represents an
    informative exchange in a classical gossip process. A bound on the number of such exchanges
    gives the following circuit-size bound.
    \begin{corollary}[Size bound for forward pivot removal] \label{cor:forward_gate_bound}
        An $n$-qubit circuit consisting of two-qubit gates with good forward lightcone structure
        has at most $\binom n2$ gates. If each of its $d$ layers is a perfect matching, then
        $d\le n-1$.
    \end{corollary}
    \begin{proof}
        Fix a complete front-pivot removal order and process the gates in reverse order,
        starting at the circuit outputs. Assign each qubit a distinct secret labeled by its
        output wire. At each gate, let its two participating qubits exchange all their secrets,
        replacing both knowledge sets by their union.

        Immediately before processing a gate, each participant knows exactly the output labels
        in the forward lightcone of its wire immediately after that gate. This follows
        recursively from the union rule for lightcones. By the preceding characterization, the
        two sets differ, so at least one participant learns a new secret. Thus every call is
        productive. By~\cite[Theorem~1.4]{brouwer2015lossy}, such a schedule has at most
        $\binom n2$ calls, giving the gate bound.

        For perfect-matching layers, the gate count is $nd/2$, so $nd/2\le\binom n2$ and hence
        $d\le n-1$.
    \end{proof}
    \noindent This bound applies to every two-qubit architecture admitting complete forward
    pivot removal, independently of randomness or signal propagation. For layered circuits,
    it limits such removal to depth $d\le n-1$. Our random ensemble loses its pivots much
    earlier, near $\log_2 n+\log_2\ln n$ (\Cref{thm:lightcone_threshold}).

    No analogous static characterization of good lightcone structure holds for
    alternating forward--backward learning (\Cref{prop:front_back_front}), where removing back
    gates changes forward lightcones. In particular, good forward lightcone structure excludes
    consecutive gates for which the support of the earlier gate is contained in that of the later
    gate. However, learning from both the front and back is strictly
    more powerful, and can modify lightcones dramatically at intermediate stages. We
    can show that this is indeed the case with a simple example.
    \begin{proposition} \label{prop:front_back_front}
        There exists a unitary quantum circuit $C$ with a back pivot but no front pivots.
        Additionally, upon removing a back pivot gate, this circuit obtains front pivots.
    \end{proposition}
    \begin{proof}
        Consider the following unitary circuit of four qubits and five gates.
        \begin{center}
        \begin{quantikz}[transparent]
            \lstick{0} & \gate[2]{A}  & \gate[3, label style={yshift=0.2cm}]{B}  &
            & \gate[2]{D} & \\
            \lstick{1} &              & \linethrough & \gate[3, label style={yshift=0.2cm}]{C}
            &             & \\
            \lstick{2} &              &              &  \linethrough & \gate[2]{E} & \\
            \lstick{3} &              &              &               &             & 
        \end{quantikz}
        \end{center}
        Consider gate $A$. It is not a pivot gate, because the forward lightcones of qubits $0$
        and $1$ at input are full, with or without $A$. However, upon removing gate $E$, qubit 3 is
        on the causal boundary of qubit 0 at input. Thus, $A$ is a pivot when $E$ is removed.
        Additionally, $E$ is a pivot gate in $C$ because qubit 2 is in the backward lightcone of 3
        (from output) only if $E$ is present.
    \end{proof}
    \noindent For a slightly more interesting example than that of the above proof,
    see~\Cref{fig:forward-back-forward-example}. This circuit does not have good forward or
    backward lightcone structure, but does have good lightcone structure when considering a
    sequence of pivots alternating from the front and back. For iterative forward-backward
    learning, verifying good lightcone structure seems to admit no analytical shortcuts. However,
    one can resort to simulating the full process of removing pivots, a process that is
    classically efficient for polynomially-sized circuits.
    \begin{figure}
    \begin{subfigure}{\textwidth}
        \centering
        \includegraphics[width=0.8\textwidth]{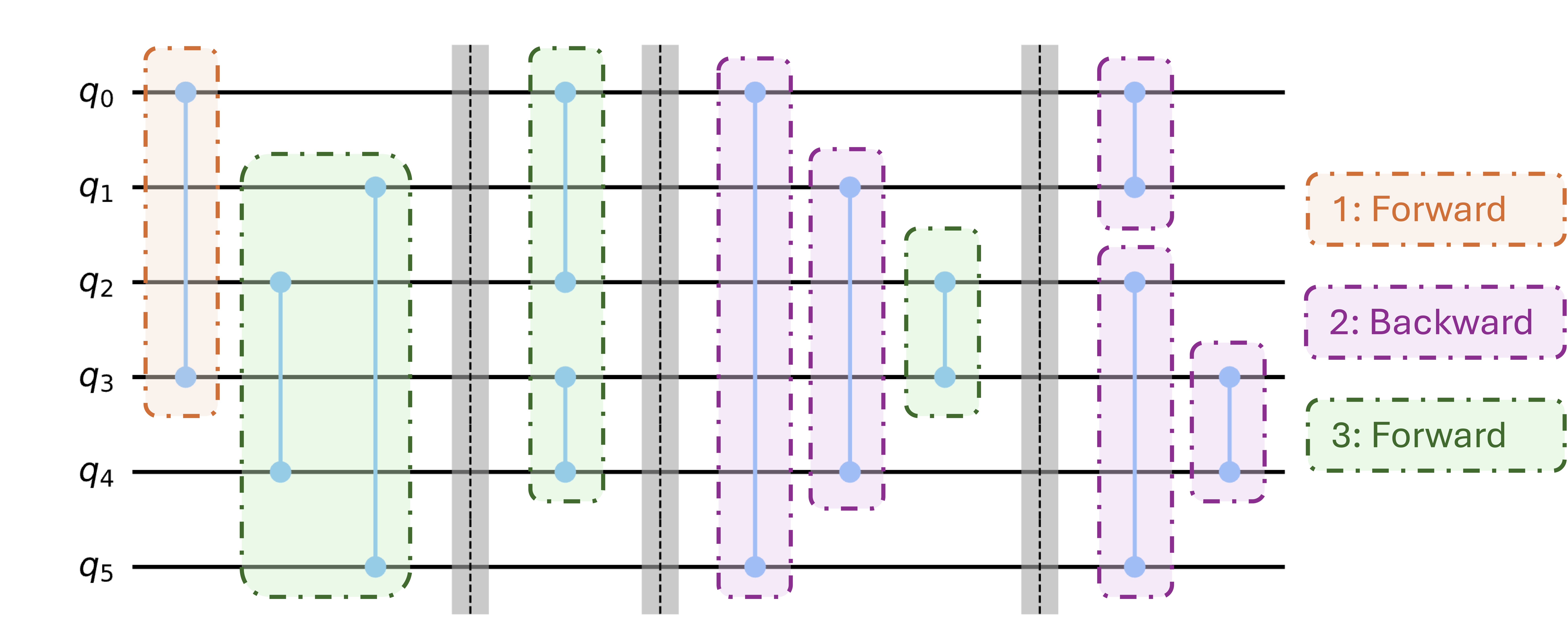}
        \caption{Front-first pivot removal. On the first forward pass (orange), only the
        first $(q_0, q_3)$ gate is a pivot. On a subsequent backward pass (purple), the five gates
        in purple can be removed, but no more. The remaining gates can now be removed on a final
        forward pass (green).}
    \end{subfigure}

    \bigskip
    
    \begin{subfigure}{\textwidth}
        \centering
        \includegraphics[width=0.8\textwidth]{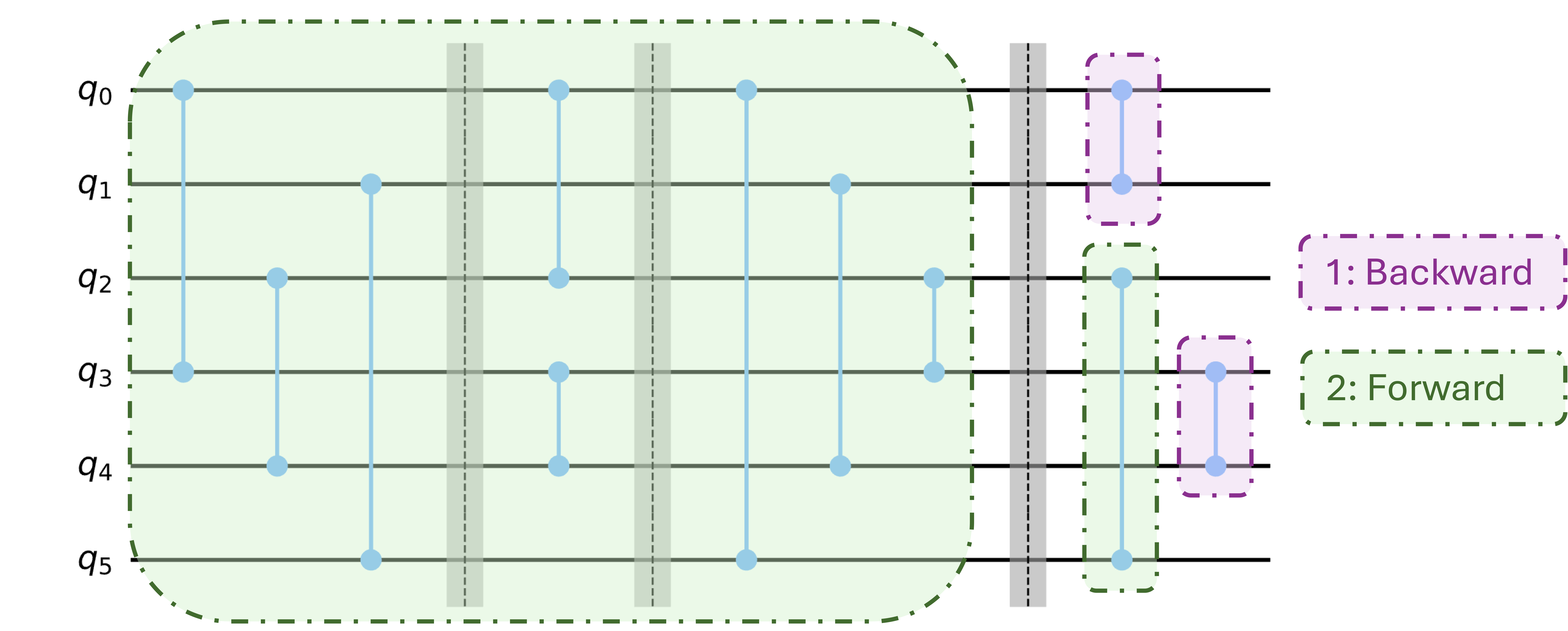}
        \caption{Back-first pivot removal. A backward pass (purple) removes the two
        gates in purple, but no others. On a subsequent forward pass (green), all remaining gates
        can be removed. }
    \end{subfigure}
    \caption{A six-qubit, four-layer circuit with two-qubit gates that has good lightcone
    structure, but neither good forward nor good backward lightcone structure. Removing pivots
    only from the front or only from the back leaves a nonempty residual circuit. Alternating
    front and back removal clears the circuit: each pass removes available pivots until none
    remain on that side.}
    \label{fig:forward-back-forward-example}
\end{figure}
    
    We now discuss the second identified property for learnability. Even if $C$ has good lightcone
    structure, it is clearly necessary that when qubits $q,q'$ are causally connected in the sense
    of paths in $C$, the corresponding causal \emph{effects} according to the unitary gates must
    be measurable. For example, if all gates are close to identity, or the circuit is deep,
    detecting tiny changes in the input of a single qubit may become intractable.
    
    We quantify the ability of a circuit to propagate changes at input to the output through a
    lower bound on distinguishability for single-qubit channels induced by the circuit.
    \begin{definition} \label{def:good_signal_propagation}
        Let $L > 0$. A quantum channel $\cal{E}$ on a Hilbert space $\cal{H}$ is $\bm{L}$\textbf{-invertible}
        with respect to a norm $\norm{\cdot}$ on operators if 
        for all input states $\rho, \sigma \in D(\cal{H})$,
        \begin{equation}
            \norm{\cal{E}(\rho) - \cal{E}(\sigma)} \geq L \norm{\rho - \sigma}.
        \end{equation}
    \end{definition}
    \noindent Observe that $\cal{E}$ is invertible iff it is $L$-invertible for some $L>0$. For
    unitarily invariant $\norm{\cdot}$, unitary channels are 1-invertible. Since quantum channels
    are contractive, $L \in [0,1]$ with respect to the trace norm. The constant $L$ bounds how
    much distinguishability the channel preserves.

    This invertibility definition above quantifies the ability to distinguish output states of a
    channel given distinct inputs. For learning shallow circuits, we ask that $C$ possess
    $L$-invertible single-qubit reduced channels, for $L$ ``not too small.''
    \begin{definition}[Signal propagation] \label{def:signal_propagation}
        Let $C$ be a unitary quantum circuit. For an input qubit $q$, an output qubit $q'$,
        and a fixed state $\tau$ on the other inputs, define the induced single-qubit channel
        \begin{equation}
            \cal{E}_{C,q\to q'}^\tau(\rho) \coloneqq
            \Tr_{Q\setminus\{q'\}}\left[U_C(\rho\otimes \tau)U_C^\dagger\right].
        \end{equation}
        For a specified collection of triples $(q,q',\tau)$, we say that $C$ has
        \textbf{signal propagation} $L(d)\in[0,1]$ if each specified channel is
        $L(\dist(q,q'))$-invertible. For random circuits, the collection is fixed before
        sampling the gates.
    \end{definition}
    \noindent Of course, $L(\infty) = 0$. We are most interested in the scaling of $L(d)$, and
    typically we should expect $L(d)\to 0$ as $d \to \infty$. The slower $L$
    decays, the easier one can detect changes at output due to state modifications at input,
    improving the detection of gate factorizations. Demonstrating $L$-invertibility must be done
    for specific cases or ensembles. For example, the following result
    from~\citeauthor{fefferman2024anti} bounds signal decay for each background fixed before
    drawing the gates.
    \begin{theorem}[Thm 1.1 of~\cite{fefferman2024anti}]
        Let $C$ be a random quantum circuit with a fixed architecture, where each gate is a
        $k$-qubit independent Haar-random unitary. Let $(q,q')$ be a pair of input and output
        qubits at gate distance $d$. Before drawing $C$, arbitrarily fix its
        other inputs, and let $\cal{E}_C$ be the channel that maps $q$ to $q'$ under this
        circuit. Then for every $\gamma > 0$, with probability at least $1-\gamma$ over $C$ the
        following holds: For every two single-qubit states $\rho$ and $\rho'$,
        \begin{equation}
            \norm{\cal{E}_C(\rho) - \cal{E}_C(\rho')}_F \geq
            (2^{-d}\gamma)^{c_k}\norm{\rho - \rho'}_F
        \end{equation}
    where $c_k > 0$ is a constant that depends only on $k$, and $\norm{\cdot}_F$ is the Frobenius
    norm.
    \end{theorem}
    \noindent The theorem also supplies signal propagation for a prescribed collection of
    channels, as required by~\Cref{def:signal_propagation}. Fix $m$ induced channels, with their
    retained circuits, input-output pairs, and background input states chosen independently of
    the sampled gates. Applying the theorem with $\gamma=\delta/m$ and taking a union bound
    gives simultaneous signal propagation
    \begin{equation}
        L(d)=\left(2^{-d}\frac{\delta}{m}\right)^{c_k}
    \end{equation}
    with probability at least $1-\delta$. The corresponding bounds for our discretized gate
    ensemble are established in~\Cref{app:discretized_haar_learning}.

    When a quantum circuit $C$ has good signal propagation, it is straightforward to
    experimentally test whether two qubits $q, q'$ are connected from input to output: simply
    toggle $q$ between $\ket{0}$ and $\ket{1}$, and perform tomography on $q'$ to accuracy
    $O(L(d))$. The challenge that arises is when, in the process of trial factorizations, an
    approximate factorizing gate $G'$ does not fully succeed, yet passes scrutiny of the
    tomography check. To handle this, we borrow analysis from~\cite{fefferman2024anti} to
    characterize the relationship between factored gates and tomographic influence. Define the
    \emph{factorization discrepancy} on $\cal{U}(2^k)$ by
    \begin{equation} \label{eq:factorization_discrepancy}
        d_\otimes(G, G') \coloneqq \min_{U_1, U_2} d_\diamond(G, (U_1\otimes U_2) G')
    \end{equation}
    where the $\min$ is taken over all unitary factors $U_1,U_2$ acting across any nontrivial
    bipartition of the $k$ input qubits. Both parts must be nonempty, but singleton parts are
    allowed. In particular,
    \begin{equation}
        d_\otimes (G,G') = 0
    \end{equation}
    iff $G G'^\dagger = U_1 \otimes U_2$ across some nontrivial bipartition. This quantity
    measures departure from a successful factorization. We call it a \emph{discrepancy} rather
    than a distance because it can vanish for distinct gates and, for $k>2$, need not satisfy
    the triangle inequality.
    
    In the simple $k = 2$ setting, $d_\otimes = 0$ implies factorization into single-qubit gates.
    The following lemma from~\cite{fefferman2024anti} relates a gap promise on the distance from
    product gates ($d_\otimes \geq \epsilon$) to an ability to detect influence via tomography.
    \begin{lemma}[Contrapositive of Lemma 6.5 of \cite{fefferman2024anti}]
    \label{lem:tensor_bound_to_tomography}
        Suppose $U \in \cal{U}(4)$ satisfies $d_\otimes(U, I\otimes I) \geq \epsilon$. Then there
        exist $\sigma_1 \in \{X, Y, Z\}$ and $\sigma_2 \in \{I, X, Y, Z\}$ such that
        \begin{equation}
            \norm{\Tr_1(U \sigma_1 \otimes \sigma_2 U^\dagger)}_F \geq \frac{\epsilon^2}{400}.
        \end{equation}
    \end{lemma}
    \noindent In learning applications, $\sigma_1$ is a difference of Pauli eigenstates, while
    $\sigma_2$ represents pieces of a full qubit state in the Pauli basis (identity included).
    Thus,~\Cref{lem:tensor_bound_to_tomography} translates to a lower bound on the change in qubit
    2 on output. This change can subsequently propagate through a circuit with good signal
    propagation, and be detected at output.

    Let's now return to the question of approximate factorizations. The simplest resolution is to
    implement a gap promise so that any unsuccessful factorization $G G'^\dagger$ is far from a
    factorization. This allows for simple hypothesis testing to reliably decide which case is
    true. To formalize this, we define the notion of ``well-spaced'' discrete sets with respect to
    a discrepancy.
    \begin{definition} \label{def:well_spaced}
        Let $M$ be a set with a nonnegative discrepancy $d:M\times M\to[0,\infty)$.
        For $\epsilon>0$, a subset $S \subseteq M$ is
        $\mathbf{\epsilon}$\textbf{-spaced} if for all $s, s' \in S$, $d(s,s') < \epsilon$ implies
        $s = s'$. We say $S$ is \textbf{well-spaced} if it is $\epsilon$-spaced for some $\epsilon
        >0$.
    \end{definition}
    \noindent Well-spacedness gives a uniform positive lower bound on the discrepancy between
    distinct elements. Unitary invariance gives
    $d_\otimes(G,G')=d_\otimes(GG'^\dagger,I)$, so taking $\cal{G}$ to be
    $\epsilon$-spaced with respect to $d_\otimes$ ensures that either (a) $G G'^\dagger=I$, in
    which case changes in $q$ produce no detectable change in $q'$, or (b)
    $d_\otimes(GG'^\dagger,I)\ge\epsilon$. With sufficiently precise tomography of the output,
    these scenarios are distinguishable in the two-qubit setting.

    We now combine good forward lightcone structure, signal propagation, and a well-spaced
    gate family to obtain an exact proper-learning guarantee for two-qubit circuits. The query
    bound makes explicit its dependence on circuit size, signal decay, and gate separation.
    
    \begin{theorem} \label{thm:learning_forward_general}
        Let $C$ be a two-local, depth $D$ circuit with good forward lightcone structure
        (\Cref{mydef:good_lightcones}). Fix an architecture-determined complete front-pivot
        order $G_1,\ldots,G_m$ and witnesses $(a_j,z_j)$, with
        $\partial G_j=\{a_j,b_j\}$. Assume that each retained circuit
        $R_j\coloneqq C\setminus\{G_1,\ldots,G_j\}$ has signal propagation $L$
        (\Cref{def:signal_propagation}) in Frobenius norm for the channel from $b_j$ to
        $z_j$ with every other input in $\ket{0}$, where
        $L:\{0,1,2,\ldots\}\to(0,1]$ is nonincreasing and $L(0)=1$. Finally,
        suppose the gate set $\cal{G}$ is given as a finite classical list and is
        $\epsilon$-spaced with respect to $d_\otimes$, and the learner can implement each
        listed gate's inverse exactly on any specified pair of qubits. Then, for every
        $\eta\in(0,1)$,
        \Cref{prob:black_box_learning} can be solved exactly, with probability at least
        $1-\eta$ over the measurements, with
        \begin{equation}
            O\left(\frac{\abs{\cal{G}}nD}{L(D)^2 \epsilon^4}
            \log\frac{2\abs{\cal{G}} n D}{\eta}\right)
        \end{equation}
        queries to $U_C$.
    \end{theorem}
    \begin{proof}
        Consider $C$ at some (possibly intermediate) stage of the circuit learning protocol, where
        the preceding gates in the prescribed order have been successfully removed. Let
        $G=G_j$, $q_1 =a_j$, $q_2=b_j$, and $q'=z_j$. This gate is a front pivot witnessed by
        $(q_1,q')$. We apply the \texttt{FactorizeFrontGate}
        protocol for $k = 2$, using inputs $\{q_1, q_2\} = \partial G$. The protocol proceeds by
        iterating over $G'$ until a successful inverse of $G$ is found. Let us verify that, with
        suitable choice of parameters ($\epsilon_f, \epsilon_t$) in \texttt{FactorizeFrontGate}, a
        correct inverse will be identified.

        Suppose $G' \neq G$. Then, because $\cal{G}$ is $\epsilon$-spaced by assumption,
        $d_\otimes(G, G') \geq \epsilon$. Hence, by \Cref{lem:tensor_bound_to_tomography}, our
        \texttt{FactorizeFrontGate} protocol will, during the Pauli iterations, find some
        $\sigma_1 \in \{X,Y,Z\}$ and $\sigma_2\in \{I, X, Y, Z\}$ such that
        \begin{equation}
            \norm{\Tr_1(G G'^\dagger \sigma_1 \otimes \sigma_2 G' G^\dagger )}_F \geq
            \frac{\epsilon^2}{400}.
        \end{equation}
        Suppose $\sigma_2 = I$. Then our protocol sets $\rho_2 = I/2$, and toggles between
        $\rho_1^{\pm} = (I \pm \sigma_1)/2$. Let $\rho_2^{\pm}$ be the corresponding state of
        $q_2$ following application of $U \coloneqq G G'^\dagger$. We have
        \begin{align}
        \begin{aligned}
            \norm{\rho_2^+ - \rho_2^-}_F &= \norm{\Tr_1(U \rho_1^+\otimes(I/2) U^\dagger) -
            \Tr_1(U \rho_1^- \otimes (I/2) U^\dagger)}_F \\
            &= \norm{\frac12 \Tr_1(U (\sigma_1 \otimes I) U^\dagger)}_F \\
            & \geq \frac{\epsilon^2}{800}.
        \end{aligned}
        \end{align}
        Now consider the case $\sigma_2 \in \{X,Y,Z\}$. Define, for example,
        \begin{equation}
            \rho_2^{+-} \coloneqq \Tr_1\left(U \frac{I+\sigma_1}{2}\otimes
            \frac{I-\sigma_2}{2}U^\dagger\right).
        \end{equation}
        Define $\rho_2^{++}, \rho_2^{-+}, \rho_2^{--}$ analogously. We now show that either
        \begin{equation} \label{eq:one_or_other}
            \norm{\rho_2^{++} - \rho_2^{-+}}_F \geq \frac{\epsilon^2}{800}\quad \text{or} \quad
            \norm{\rho_2^{+-} - \rho_2^{--}}_F \geq \frac{\epsilon^2}{800}.
        \end{equation}
        Indeed, a direct calculation shows that $(\rho_2^{++} - \rho_2^{-+}) - (\rho_2^{+-} -
        \rho_2^{--}) = \Tr_1( U \sigma_1 \otimes \sigma_2 U^\dagger)$. Hence,
        \begin{align}
        \begin{aligned}
            \frac{\epsilon^2}{400} &\leq \norm{\Tr_1 (U \sigma_1\otimes \sigma_2 U^\dagger)}_F \\
            &= \norm{(\rho_2^{++} - \rho_2^{-+}) - (\rho_2^{+-} - \rho_2^{--})}_F \\
            &\leq 2 \max\left\{\norm{\rho_2^{++} - \rho_2^{-+}}_F, \norm{\rho_2^{+-} -
            \rho_2^{--}}_F\right\}
        \end{aligned}
        \end{align}
        which implies~\Cref{eq:one_or_other}. Let $b \in \{0,1\}$ denote the sign of the Pauli
        eigenstate for $q_2$, at input, such that the lower bound holds, i.e.
        \begin{equation}
            \rho_2 = \frac{I + (-1)^b \sigma_2}{2}.
        \end{equation}

        Thus, \texttt{FactorizeFrontGate} eventually tests a configuration in which toggling
        $q_1$ between $\rho_1^\pm \coloneqq (I\pm\sigma_1)/2$, with the input state of $q_2$
        fixed, produces reduced output states of $q_2$ after $U$ whose Frobenius distance is
        at least $\epsilon^2/800$. Let us now
        denote the associated states of $q_2$ after $G G'^\dagger$ by $\rho_2^+, \rho_2^-$. These
        states then pass through the retained circuit $R_j$. Let $\cal{F}_{C\setminus G}$
        be its channel from $q_2$ to $q'$ with every other input in $\ket{0}$.
        Since $q_1$ cannot reach $q'$ in $R_j$, the output on $q'$ depends only on the $q_2$
        marginal, even when $q_1$ and $q_2$ are correlated. We may therefore use the channel
        with $q_1$ fixed to $\ket{0}$. Similarly, let
        $\cal{E}_{C \cup G'^\dagger}$ be the single-qubit channel from $q_1$ to $q'$ induced by $C$
        with $G'^\dagger$ added. We have the relation
        \begin{equation}
            \cal{E}_{C\cup G'^\dagger}(\rho_1) = \cal{F}_{C\setminus G}(\rho_2),
        \end{equation}
        where $\rho_1$ is the input state of $q_1$ and $\rho_2$ is the state of $q_2$ following
        $U$. Having given a lower bound on the ``effect'' on $\rho_2$ by changing $\rho_1$, we can
        now lower bound the change in output. Let $\pi^\pm \coloneqq \cal{E}_{C\cup
        G'^\dagger}(\rho_1^\pm)$. Then
        \begin{align}
        \begin{aligned}
            \norm{\pi^+ - \pi^-}_F &= \norm{\cal{E}_{C\cup
            G'^\dagger}(\rho_1^+)-\cal{E}_{C\cup G'^\dagger}(\rho_1^-)}_F \\
            &= \norm{\cal{F}_{C\setminus G}(\rho_2^+) - \cal{F}_{C\setminus G}(\rho_2^-)}_F \\
            &\geq L(D) \norm{\rho_2^+ - \rho_2^-}_F \\
            &\geq L(D) \frac{\epsilon^2}{800}
        \end{aligned}
        \end{align}
        where the third line uses propagation for this fixed-input channel and
        $\dist_{R_j}(q_2,q')\le D$. The bound holds for every pair of partner states,
        including those depending on the trial gate. A zero-length connection is an
        untouched wire and has propagation constant $1$.

        On the other hand, if $G' = G$, then for any fixed state on $q_2$ we have $\norm{\pi^+ -
        \pi^-}_F = 0$. Thus, by performing single-qubit tomography within precision $O(L(D)
        \epsilon^2)$, one can distinguish between these two cases with controlled failure probability.

        The procedure uses only product-state preparations consisting of single-qubit
        Pauli eigenstates, maximally mixed states, and computational-basis states, followed by
        known corrections and single-qubit Pauli measurements. The corrections may entangle the
        register before it enters the oracle. Each fresh execution of the corrected circuit,
        including repetitions for tomography, uses one query to $U_C$.

        The cost of the
        single-qubit tomography is $O(\epsilon^{-4} L(D)^{-2})$ queries to
        $\pi^\pm$~\cite{nielsen2010quantum}, hence that many calls to $U_C$. This tomography is
        repeated $O(\abs{\cal{G}} \abs{C})$ times, where $\abs{C} = O(n D)$ is the number of gates.
        More explicitly, put $\Delta=L(D)\epsilon^2/800$, estimate each output state to
        Frobenius error at most $\Delta/8$, and reject at estimated difference $\Delta/2$.
        Correct candidates pass, and each incorrect candidate fails some test. Allocate
        failure probability $O(\eta/(\abs{\cal{G}}nD))$ to each estimate. The resulting
        logarithmic overhead ensures that all tests succeed with probability at least $1-\eta$.
        On this event the accepted gate is the true gate at every step, by induction, so the
        implemented retained circuits are exactly the prescribed $R_j$. This proves the bound.
    \end{proof}

\subsection{Scope and limitations}
    \Cref{thm:learning_forward_general} gives a forward-learning guarantee for a well-spaced
    two-qubit gate family. We now consider backward learning, continuous gate families, and
    gates acting on more than two qubits.

    Our forward-learning proof relies on precise control of the input states used in the
    factorization tests. At a back pivot, these states have already passed through the unknown
    preceding circuit, so varying the circuit inputs does not automatically provide the same
    control. A heuristic backward procedure could nevertheless vary those inputs to seek
    changes at the pivot that distinguish candidate factorizations. Such variation may suffice
    in practice, although we do not establish a guarantee that it supplies the required tests.
    With additional query access to $U_C^\dagger$, backward learning can instead be formulated
    as forward learning of the inverse circuit.
    
    For continuous gate families, the gap between distinct candidates is unavailable, and local
    factorization tests generally yield approximate product decompositions rather than exact
    inverses. In the two-qubit setting, the resulting single-qubit factors can be absorbed into
    subsequent gates when the gate family permits, as for $\cal{G}=\cal{U}(4)$; terminal factors
    can be learned by tomography. With suitable control of approximation errors, this realizes
    \texttt{HandleFactors} while preserving the prescribed layout and gate family.

    Controlling these errors introduces an additional cost. For logarithmic-depth circuits with
    Haar-random gates, the analysis of~\cite{fefferman2024anti} gives quasi-polynomial learning
    runtimes. It remains unclear whether this cost is intrinsic to local inversion or whether
    improved error control can reduce it.

    For gates acting on more than two qubits, recursive splitting requires causal boundaries on
    smaller supports. A boundary on the original gate support does not ensure that these remain
    available after splitting. A learning guarantee for larger gates also requires suitable
    factorization tests and factor handling within the prescribed layout and gate family.
    The following condition guarantees a forward causal boundary on every subset containing
    at least two qubits.
    \begin{definition}[Pairwise-distinct lightcones]
        \label{def:pairwise_distinct_lightcones}
        For a gate $G$, let $L_G(a)$ denote the output qubits reachable from wire $a$
        immediately after $G$. We say that $G$ has \textbf{pairwise-distinct (forward)
        lightcones} if
        \begin{equation}
            L_G(a)\neq L_G(b)
            \qquad\text{for all distinct }a,b\in\partial G.
        \end{equation}
        A circuit satisfies this condition if every gate does.
    \end{definition}
    To see its role, fix the circuit following $G$ and consider a residual factor on a
    support $T\subseteq\partial G$. Its forward causal boundary on wire $a\in T$ is
    \begin{equation}
        \left(\bigcup_{b\in T}L_G(b)\right)\setminus L_G(a).
    \end{equation}
    If the cones are pairwise distinct and $|T|\ge2$, this boundary is nonempty for some
    $a\in T$. Conversely, if $L_G(a)=L_G(b)$, a factor on $\{a,b\}$ has no forward causal
    boundary. Thus, pairwise distinctness is equivalent to the availability of a forward
    causal boundary on every residual factor acting on two or more qubits. Since the condition is
    inherited by subsets, it supplies the pivots needed throughout exact recursive
    splitting into disjoint supports.
    
    For two-qubit gates, this condition coincides with the pivot condition when the gate
    is at the front. For larger gates it is stronger: $L_G(a)=L_G(b)\neq L_G(c)$ permits
    a pivot on $\{a,b,c\}$ but none on a residual factor acting on $\{a,b\}$.
    
    Pairwise distinctness persists as outer gates are removed. Front removal leaves the output
    lightcones of each remaining gate unchanged, while back removal can change the cones and
    their witnesses without making distinct cones equal. Indeed, restoring a back gate maps
    equal lightcones to equal lightcones, so removing it cannot destroy distinctness.
    
\section{Causal boundaries in layered circuits} \label{sec:random_circuit_ensemble}

We now describe how lightcone evolution determines the causal boundaries available in layered
circuits, before applying this description to random all-to-all architectures.
For local inversion, the important question is whether lightcones from different wires remain
distinct: an output reached from one wire but not another can witness a causal boundary.
We therefore follow two lightcones together and count the qubits reached by exactly one of them.

Consider a depth-$d$ layered circuit, with $P_1,\ldots,P_d$ denoting the partitions of its qubits
into gate supports. Fix a gate $G$ in layer $\ell$, and choose distinct qubits
$a,b\in\partial G$ immediately after the gate.
A lightcone expands across a layer by including all qubits of every gate it reaches.
For a partition $P$ representing this layer and a set of qubits $V$, define the
\emph{overlap set}
\begin{equation}
    O_P(V)\coloneqq\{B\in P:B\cap V\neq\emptyset\}.
\end{equation}
This is the set of blocks reached by $V$, so the lightcone update is
\begin{equation} \label{eq:layer_size}
    \Phi_P(V)\coloneqq\bigcup_{B\in O_P(V)}B,
    \qquad
    |\Phi_P(V)|=\sum_{B\in O_P(V)}|B|.
\end{equation}
Also, $V\subseteq\Phi_P(V)$, and each reached block contains at least one qubit of $V$.
Thus, if every block has size at most $k$,
\begin{equation} \label{eq:layer_growth}
    |V|\le|\Phi_P(V)|\le k|O_P(V)|\le k|V|.
\end{equation}
For blocks of size exactly $k$, the size identity reduces to $|\Phi_P(V)|=k|O_P(V)|$.
Let $A_j$ and $B_j$ denote the forward lightcones of $a$ and $b$ through the first $j$ layers
following $G$. Before any subsequent gates act, each lightcone contains only its root:
$A_0=\{a\}$ and $B_0=\{b\}$. Their layer-by-layer evolution is given by
\begin{equation}
    A_{j+1}=\Phi_{P_{\ell+j+1}}(A_j),
    \qquad
    B_{j+1}=\Phi_{P_{\ell+j+1}}(B_j),
    \qquad 0\le j<d-\ell.
\end{equation}
These updates apply to any fixed sequence of layer partitions.
At $m=d-\ell$, $A_m$ and $B_m$ are the forward lightcones from the two wires to the circuit output,
so $L_G(a)=A_m$ and $L_G(b)=B_m$ in the notation of~\Cref{def:pairwise_distinct_lightcones}.

Our main quantity of interest is the size of their symmetric difference,
\begin{equation}
    H_j(a,b)\coloneqq|A_j\triangle B_j|
    =|A_j\setminus B_j|+|B_j\setminus A_j|.
\end{equation}
We write $H_j$ when the pair of roots is fixed.
Although each cone grows monotonically, $H_j$ need not do so: it may first grow as the cones
spread, then shrink as they overlap. Once the two cones become equal, however, they remain
equal under every subsequent layer:
\begin{equation} \label{eq:lightcone_equality_persists}
    H_j=0
    \quad\Longleftrightarrow\quad A_j=B_j
    \quad\Longrightarrow\quad H_{j+r}=0,
    \qquad 0\le r\le m-j.
\end{equation}
The pairwise-distinct lightcone condition of~\Cref{def:pairwise_distinct_lightcones} is therefore
\begin{equation} \label{eq:pairwise_difference_condition}
    H_m(a,b)>0
    \qquad\text{for all distinct }a,b\in\partial G.
\end{equation}
As discussed in~\Cref{sec:algorithm}, this condition ensures a forward causal boundary on
every residual factor acting on two or more qubits.
Conversely, if all lightcones immediately after $G$ coincide, then $G$ has no forward
causal boundary.

\subsection{Random layered circuits} \label{subsec:lightcone_analysis}

For the random all-to-all layered ensemble defined below, with fixed gate locality $k$
and asymptotically large $n$, we prove that the regimes of pairwise-distinct lightcones
and absence of outer pivots are
separated by a constant-width window around the threshold
\begin{equation}
    t=t(n,k)=\log_k n+\log_k\ln n.
\end{equation}
Though our results are asymptotic,
we expect this picture to remain qualitatively useful at smaller circuit sizes. When $d(n)$ lies below this
window, the circuit has the causal boundaries needed for recursive factorization. Above the window,
neither forward nor backward local inversion has an initial pivot gate for use in a
local-inversion learning protocol. The leading $\log_k n$ term reflects the fact that
lightcones can grow by at most a factor of $k$ per layer. The additional $\log_k\ln n$ term
reflects the slower approach to full coverage: random gate placement makes no purposeful attempt
to connect qubits inside a lightcone to those outside it, and therefore does not maximize
lightcone growth at each layer.

To isolate the loss of causal boundaries caused by lightcone growth, we exclude repeated
interactions between consecutive layers. For two-qubit gates, this simply forbids applying gates
to the same pair twice in succession. For larger gates, we require that gate supports in
consecutive layers overlap in at most one qubit.
\begin{definition}[Single-overlap layered ensemble]
\label{def:random_layered_architecture}
    Fix $k\ge2$ and let $n\ge k^2$ be divisible by $k$. The \textbf{single-overlap layered ensemble}
    of depth $d$ consists of partitions $P_1,\ldots,P_d$ of $[n]$ into blocks of size $k$, sampled
    as follows. Choose $P_1$ uniformly from all such partitions. Given $P_\ell$, choose
    $P_{\ell+1}$ uniformly among partitions satisfying
    \begin{equation}
        |A\cap B|\le1
        \qquad\text{for every }A\in P_\ell,\ B\in P_{\ell+1}.
    \end{equation}
    Each block is the support of one gate. We call the requirement that
    gates in consecutive layers share at most one qubit the
    \textbf{single-qubit overlap condition}. Equivalently, no qubit pair
    shares a gate in consecutive layers.
\end{definition}
\noindent The condition $n\ge k^2$ is necessary to have single-qubit overlap.
Although consecutive layers are dependent, each layer is uniformly distributed over
all partitions into $k$-qubit blocks, and reversing the layer order preserves the circuit
distribution.
For two roots on the support of $G$, the single-qubit overlap condition places them in
distinct gates in the next layer, giving $H_1=2k$ whenever a subsequent layer is present.

Another natural model of an all-to-all random circuit dispenses with the layered
structure: choose the support of each of $N$ gates independently and uniformly from the $k$-qubit
subsets, then compose the gates in sequence. One expects similar qualitative behavior in lightcone
growth, but leave this to future work.

The single-qubit overlap condition can be enforced by rejection sampling: draw a uniform
partition and accept it only if it shares no qubit pair with the preceding layer. Write $P\sim Q$
when partitions $P$ and $Q$ satisfy this compatibility condition. The following lemma shows that,
for fixed $k$, the acceptance probability stays bounded away from zero as the system grows.
\begin{lemma}[Acceptance probability] \label{lem:acceptance}
    Fix $k\ge2$ and let $n\ge k^2$ be divisible by $k$. For any fixed partition $P$ of $[n]$ into
    blocks of size $k$, let $Q$ be an independent, uniformly random such partition. The probability
    $p_{n,k}\coloneqq\Pr(P\sim Q)$ is strictly positive, does not depend on $P$, and
    \begin{equation}
        \lim_{n\rightarrow \infty} p_{n,k} = \exp\left(-\frac{(k-1)^2}{2}\right)
    \end{equation}
    where the limit is taken in multiples of $k$. In particular, there is a $p_*(k) > 0$ such that
    $p_{n,k}\ge p_*(k)$ for all such $n$.
\end{lemma}
\noindent A proof is given in~\Cref{app:delayed_proofs}. Because the expected number of proposals tends to
$\exp((k-1)^2/2)$, a depth-$d$ architecture can be sampled in expected time $O_k(nd)$.
The single-qubit overlap condition ensures that the compression step handled
in~\Cref{sec:algorithm} is unnecessary, simplifying our analysis.

The following theorem establishes pairwise-distinct lightcones below the threshold and
rules out even ordinary pivots above it.
\begin{theorem}[Lightcone threshold] \label{thm:lightcone_threshold}
    Fix $k\ge2$ and let $n\to\infty$ through multiples of $k$. For the single-overlap ensemble
    of~\Cref{def:random_layered_architecture} with depth $d=d(n)$, put
    $t=t(n,k)=\log_k n+\log_k\ln n$.
    There is a constant $\Delta_k>0$, depending only on $k$, such that:
    \begin{enumerate}
        \item If $d\le t-\Delta_k$, the forward lightcones immediately after every gate are pairwise
        distinct on its support, with probability $1-o(1)$.
        \item If $d\ge t+\Delta_k$, the circuit has neither a front pivot nor a back pivot, with
        probability $1-o(1)$.
    \end{enumerate}
\end{theorem}
\noindent For $d\le t-\Delta_k$, \Cref{thm:lightcone_threshold} supplies, with high probability,
the lightcone structure required by the local-inversion framework of~\Cref{sec:algorithm}.
Its pairwise-distinct lightcones provide causal boundaries throughout recursive splitting,
regardless of the gate unitaries. The remaining requirements concern signal propagation and
the implementation of local factorizations. 

For two-qubit gates, we obtain a concrete learning guarantee by choosing a finite gate family
that is both well-spaced and obtained by rounding Haar-random unitaries. Well-spacedness
provides the gap needed to identify exact inverses, while the connection to Haar measure
allows us to establish signal propagation. We use the following specialization
of~\cite[Definition~6.7]{fefferman2024anti}.
\begin{definition}[Well-spaced Haar discretization] \label{def:separated_haar_discretization}
    For $h\in(0,1/2]$, let $\cal{G}_h\subset\cal{U}(4)$ be a finite $h$-spaced set with
    respect to $d_\otimes$ (\Cref{def:well_spaced}), and let
    $f_h:\cal{U}(4)\to\cal{G}_h$ be a measurable map such that
    \begin{equation} \label{eq:haar_rounding_conditions}
        d_\otimes(U,f_h(U))\le h
        \quad\text{for every }U\in\cal{U}(4).
    \end{equation}
    The distribution $\mu_h$ is the law of $f_h(H)$ for Haar-random $H\in\cal{U}(4)$.
\end{definition}
\noindent The definition imposes two complementary conditions: distinct listed gates are
separated by at least $h$ in factorization discrepancy, and every two-qubit unitary can be
rounded to a listed gate with discrepancy at most $h$. Sampling a Haar-random unitary and
applying this rounding map assigns each gate the Haar mass of its rounding cell, rather than
choosing uniformly from the list. In~\Cref{app:discretized_haar_learning}, we control the effect
of this rounding on signal propagation. Together with the below-threshold lightcone
structure, these estimates yield the following learning guarantee.

\begin{corollary}[Proper learning below the threshold, informal] \label{cor:single_overlap_learning}
    Consider two-local circuits with layouts drawn from the single-overlap ensemble below
    the threshold window of~\Cref{thm:lightcone_threshold}. Draw the gates independently
    from a well-spaced Haar discretization $\mu_h$, with a suitable inverse-polynomial
    choice of $h$. With high probability over the layouts, gates, and measurements,
    the circuit admits exact proper learning using polynomially many queries.
    The learner is given the layout and finite gate list and can implement each listed
    inverse exactly.
\end{corollary}
\noindent A quantitative statement and proof, including the required choice of $h$,
appear in~\Cref{app:discretized_haar_learning}.

We now turn to the proof of~\Cref{thm:lightcone_threshold}, which follows the lightcones
through two depth regimes. During early growth,
over a logarithmic number of layers, we show that both a single lightcone and the symmetric
difference of a pair grow to occupy a fixed fraction of the qubits with high probability.
At later depths, we study how long two lightcones remain different as their overlap grows,
and when full coverage eliminates the causal boundaries needed for local inversion.

Write $S_j=|A_j|$ for the size of one lightcone and $R_j=n-S_j$ for its complement size.
We use the relative quantities
\begin{equation} \label{eq:relative_lightcone_sizes}
    \sigma_j\coloneqq\frac{S_j}{n},
    \qquad
    \eta_j\coloneqq\frac{H_j}{n},
    \qquad
    \rho_j\coloneqq\frac{R_j}{n}=1-\sigma_j.
\end{equation}

\subsubsection{Early depth} \label{subsubsec:early_depth}

When a lightcone is small, most of its qubits enter different gates in the next layer, each
bringing $k$ qubits into the cone. Growth slows when several qubits in a lightcone enter the same
gate in a layer, an event we call a \emph{collision}. For a set $V$ and a layer partition $P$,
we count these collisions via the collision \emph{deficit}
\begin{equation} \label{eq:collision_deficit}
    D_P(V)\coloneqq |V|-|O_P(V)|.
\end{equation}
The collision deficit measures the shortfall in the number of reached blocks from the
collision-free case, in which every vertex of $V$ reaches a different block.
For the lightcone evolution following gate $G$, write
$D_j(V)\coloneqq D_{P_{\ell+j+1}}(V)$ for the collision deficit incurred when
propagating $V$ from step $j$ to step $j+1$.

Collisions can also cause losses in the symmetric difference relative to collision-free growth.
If a gate meets both cones, all its output qubits belong to both updated cones and contribute
nothing to $H_{j+1}$.
Writing $U_j=A_j\cup B_j$, we obtain
\begin{equation} \label{eq:early_growth_bounds}
    \begin{aligned}
        S_{j+1}&=kS_j-kD_j(A_j),\\
        H_{j+1}&\ge kH_j-2kD_j(U_j).
    \end{aligned}
\end{equation}
The first identity follows from~\Cref{eq:layer_size} and the definition of $D_j$.
For the second, compare
$H_{j+1}$ with the collision-free value $kH_j$. A block containing exactly one vertex of $U_j$
causes no loss: if that vertex
belongs to exactly one cone, it contributes $k$ to $kH_j$ and produces exactly
$k$ vertices in the updated symmetric difference; otherwise, both contributions
are zero. A block containing
$r\ge2$ vertices of $U_j$ accounts for at most $kr$ in the collision-free count,
so its loss is at most $kr\le2k(r-1)$. Since $r-1$ is the block's collision
deficit, summing over reached blocks gives a total loss $k H_j - H_{j+1}$ of at most
$2kD_j(U_j)$, proving the second inequality.

The conditional collision moment bound of~\Cref{lem:conditional_collision_mgf}
accounts for compatibility between consecutive layers. Applied to the losses
in~\Cref{eq:early_growth_bounds}, it yields the following early-growth estimate.
\begin{lemma}[Early growth] \label{lem:early_growth}
    Fix $k\ge2$ and $\varepsilon>0$. There exist an integer $K\ge1$ and constants
    $0<c<1$ and $C>0$, depending only on $k$ and $\varepsilon$, such that, with
    \begin{equation} \label{eq:early_time}
        J\coloneqq\lfloor\log_k n\rfloor-K,
    \end{equation}
    the relative difference $\eta_J$ and the single-lightcone fraction $\sigma_J$
    satisfy, for sufficiently large $n$,
    \begin{equation} \label{eq:early_failure}
        \begin{aligned}
            \Pr(\eta_J<c)\le \frac{C}{n^{k-\varepsilon}}, \qquad \Pr(\sigma_J<c) \le \frac{C}{n^{k-\varepsilon}}.
        \end{aligned}
    \end{equation}
    The bounds hold uniformly over distinct roots $a,b\in\partial G$, conditional on
    the layers through $G$.
\end{lemma}
\noindent See~\Cref{app:early_depth} for a proof.

The failure probability is small enough to control all relevant root pairs in a shallow
circuit. There are $O_k(n\ln n)$ such pairs when $d=O_k(\ln n)$, so choosing
$0<\varepsilon<k-1$ makes their total early-growth failure probability vanish by a union bound.
If fewer than $J$ layers remain after a gate, extend the circuit by sampling additional layers
from the same ensemble. Any two lightcones that are equal at the original output remain equal
after the extension. The same probability bound therefore controls equality of the output cones
when fewer than $J$ layers remain.

\subsubsection{Late depth} \label{subsubsec:late_depth}

Once the lightcones occupy a fixed fraction of the qubits, collisions are common and the
early-growth bound no longer captures their behavior. The cones now overlap substantially, but
they can still differ on many outputs. We obtain complementary estimates for depths below
and above the saturation threshold. The first shows that two lightcones remain different
for a further $\log_k\ln n-O_k(1)$ layers. The second follows the qubits outside a single
cone to show that a few more layers suffice for full coverage.

Any next-layer gate lying entirely in $A_j\setminus B_j$, or entirely in $B_j\setminus A_j$,
contributes all its qubits to the next
symmetric difference. Counting these gates gives a lower bound on $H_{j+1}$, even when the
rest of the lightcones overlap.
\begin{lemma}[Persistence of distinct lightcones] \label{lem:late_persistence}
    There are constants $a_k,b_k>0$, with $b_k<1$, such that, whenever $H_j\ge4k$,
    \begin{equation} \label{eq:late_persistence}
        \Pr\left(
            \eta_{j+1}<b_k\eta_j^k
            \,\middle|\,\cal{F}_j
        \right)
        \le p_*(k)^{-1}
        \exp\left(-a_kn\eta_j^{2k}\right).
    \end{equation}
    Here $\cal{F}_j$ records the full circuit history through the $j$th layer after $G$.
\end{lemma}
\noindent The proof is given in~\Cref{app:pre_saturation}.
On successful steps, $\eta_{j+1}\ge b_k\eta_j^k$.
Starting from the constant-fraction condition $\eta_J\ge c>0$ of~\Cref{lem:early_growth},
iteration gives, provided the intervening one-step bounds hold,
\begin{equation} \label{eq:late_difference_bound}
    \eta_{J+r}\ge e^{-B_k k^r}>0,
\end{equation}
for a constant $B_k>0$ independent of $n$.
In particular, the two lightcones remain different for every such $r$.
The failure bound in~\Cref{eq:late_persistence} makes this a high-probability guarantee
for a further $\log_k\ln n$ layers minus a sufficiently large constant, simultaneously
for all root pairs on all gates below the threshold.

The persistence estimate loses useful control on the same $\log_k\ln n+O_k(1)$ scale
as lightcone saturation. As we now show, once a lightcone covers a fixed positive fraction
of the qubits, this many additional layers suffice for full coverage with high probability.
To establish the above-threshold conclusion, it is enough to show that the lightcone of every
qubit immediately after the first layer reaches all outputs. We therefore analyze the
complement of a single lightcone. A gate remains outside the next lightcone exactly when all
its inputs lie in this complement, so a complement of relative size $\rho_j$ contracts on
the scale $\rho_j^k$. The exponential-moment estimate
in~\Cref{lem:complement_moments} gives, for a next layer sampled subject to compatibility,
\begin{equation} \label{eq:late_complement_tail}
    \Pr\left(
        \rho_{j+1}>\frac{e-1}{p_*(k)}\rho_j^k+\frac{3k\ln n}{n}
        \,\middle|\,\cal{F}_j
    \right)\le n^{-3}.
\end{equation}
After a constant number of layers brings the complement below a sufficiently small fraction
of $n$, this contraction leaves only $O_k(\ln n)$ unreached qubits in
$\lceil\log_k\ln n\rceil$ further layers. One or two final layers eliminate the remainder.

\begin{lemma}[Late saturation] \label{lem:late_saturation}
    Fix $0<c<1$. There is an integer $L=L(k,c)$ such that, starting from a lightcone of size
    at least $cn$, the cone is full after
    \begin{equation} \label{eq:late_saturation_time}
        L+\lceil\log_k\ln n\rceil+2
    \end{equation}
    additional layers of the single-overlap ensemble, except with probability $o(1/n)$.
    The bound is uniform over the preceding history.
\end{lemma}
\noindent The proof, including the final elimination of the unreached qubits, appears
in~\Cref{app:post_saturation}. The $o(1/n)$ failure probability is small enough to apply the condition to all qubits after a layer, with high probability, by a union bound. Together, the two late-depth estimates explain the correction
to the leading $\log_k n$ growth time: the symmetric difference between lightcones remains
nonempty, while the unreached complement contracts, over an additional $\log_k\ln n+O_k(1)$ layers.

\subsubsection{Obtaining the theorem} \label{subsubsec:across_threshold}

We now combine the early- and late-depth estimates. Below the threshold, we need every pair
on every gate to remain distinct at the output. Above it, full coverage from the wires after
the first layer rules out all front pivots; reversing the circuit gives the same conclusion
at the back.

\begin{proof}[Proof of~\Cref{thm:lightcone_threshold}]
    Choose $0<\varepsilon=\varepsilon(k)<k-1$, fixed throughout the proof.
    Take $J,K,c$ from~\Cref{lem:early_growth}.
    Put $B_k=\ln(1/c)+\ln(1/b_k)/(k-1)$.
    For the lower bound, let $\Delta_k^->0$ be a constant to be chosen below.
    Follow a pair from one gate through
    $T_-\coloneqq\lfloor t-\Delta_k^-\rfloor$ subsequent layers, extending the sampled circuit if necessary.
    Except with probability $O_k(n^{-k+\varepsilon})$, early
    growth gives $\eta_J\ge c$.

    Conditional on this early-growth event, apply~\Cref{lem:late_persistence} while
    $\eta_{J+r}\ge n^{-1/(4k)}$. Above this level,
    \begin{equation}
        \eta_{J+r}^{2k}\ge\frac{1}{\sqrt n},
    \end{equation}
    so~\Cref{lem:late_persistence} gives $\eta_{J+r+1}\ge b_k\eta_{J+r}^k$
    except with conditional probability at most $p_*(k)^{-1}e^{-a_k\sqrt n}$.
    Since $k^t=n\ln n$ and $k^J\ge nk^{-K-1}$, the condition
    $J+r\le T_-\le t-\Delta_k^-$ gives
    \begin{equation}
        k^r
        =\frac{k^{J+r}}{k^J}
        \le\frac{k^{t-\Delta_k^-}}{k^J}
        \le k^{K+1-\Delta_k^-}\ln n.
    \end{equation}
    Choose $\Delta_k^-$ large enough that $B_k k^{K+1-\Delta_k^-}\le1/(8k)$.
    The preceding bound then gives $B_k k^r\le(\ln n)/(8k)$.
    If all preceding one-step estimates succeed, \Cref{eq:late_difference_bound} yields
    \begin{equation}
        \eta_{J+r}\ge e^{-B_k k^r}\ge n^{-1/(8k)}>n^{-1/(4k)}.
    \end{equation}
    The one-step probability bound therefore remains applicable whenever the preceding
    estimates hold. A union bound over the $O_k(\ln\ln n)$ updates, together with the
    early-growth estimate, gives
    \begin{equation} \label{eq:pair_separation_failure}
        \Pr(\eta_{T_-}=0)
        \le O_k(n^{-k+\varepsilon})+O_k(\ln\ln n)e^{-a_k\sqrt n}.
    \end{equation}

    Equality of the cones at the original output persists under extension to $T_-$.
    Thus, \Cref{eq:pair_separation_failure} applies uniformly to every gate and root pair.
    There are $(nd/k)\binom{k}{2}=O_k(n\ln n)$ such instances.
    Let $\cal{E}_{\mathrm{eq}}$ be the event that two distinct wires immediately after
    some gate have identical forward lightcones. A union bound gives
    \begin{equation}
        \Pr(\cal{E}_{\mathrm{eq}})
        \le O_k(n^{1-k+\varepsilon}\ln n)+O_k(e^{-a_k\sqrt n/2})=o(1).
    \end{equation}
    This proves the lower conclusion with offset $\Delta_k^-$.

    For the upper bound, consider the qubits immediately after the first circuit layer.
    Conditional on this layer, \Cref{lem:early_growth} brings each root cone to size at least
    $cn$ in $J$ further layers, except with probability $O_k(n^{-k+\varepsilon})$ per root.
    By~\Cref{lem:late_saturation}, all $n$ cones are therefore full after
    \begin{equation} \label{eq:all_root_saturation_time}
        T_+ \coloneqq J + L + \lceil\log_k\ln n\rceil+2
        \le t+O_k(1)
    \end{equation}
    further layers, except with total probability
    $O_k(n^{1-k+\varepsilon})+n\,o(1/n)=o(1)$.
    For $d\ge t+\Delta_k^+$ with $\Delta_k^+$ sufficiently large,
    $d-1\ge T_+$, so there is no front pivot.

    The reversed circuit has the same distribution, so the same
    argument rules out back pivots, and a union bound gives both conclusions at once.
    Additional layers preserve full coverage, so this conclusion holds for arbitrarily large
    deterministic depth sequences.
    Taking $\Delta_k=\max\{\Delta_k^-,\Delta_k^+\}$ proves the theorem.
\end{proof}

Below the threshold, the pairwise-distinct condition provides the causal boundaries needed at
each stage of exact recursive factorization. Above it, neither forward nor backward local
inversion can begin.

\subsection{Numerical experiments} 

We supplement the threshold analysis with numerical experiments in a broader setting. We sample
two-qubit layers independently, as uniformly random perfect matchings, without rejecting layouts
that contain repeated pairs. We also compare forward-only removal with a procedure that
alternates forward and backward passes. These experiments explore both the effect of dropping
the restriction on consecutive layers and the advantage of using pivots at both ends of the
circuit.

For each qubit count $n$, we sweep through depths $d$ and sample $N=1000$ circuit layouts at
each $(n,d)$. Consecutive gates acting on the same pair are treated as one composite gate,
as in~\Cref{sec:algorithm}; the reported depth is the number of sampled layers before grouping.
We compute the forward and backward lightcones, remove available outer pivots, and repeat until
no further removals are possible. Forward-only removal uses only front pivots, while the
alternating procedure also uses back pivots. A trial succeeds if the circuit is completely
removed. The reported probabilities therefore measure complete removal by the pivot procedure,
rather than execution of a learning protocol.

\Cref{subfig:full_learn_b} shows the success probability of alternating forward--backward
removal. For fixed $n$, success drops as the depth increases. As the qubit count grows in powers
of two, the transition shifts to larger depths and appears to sharpen.
\begin{figure}
    \centering
    \begin{subfigure}[]{0.45\textwidth}
        \centering
        \includegraphics[width=\textwidth]{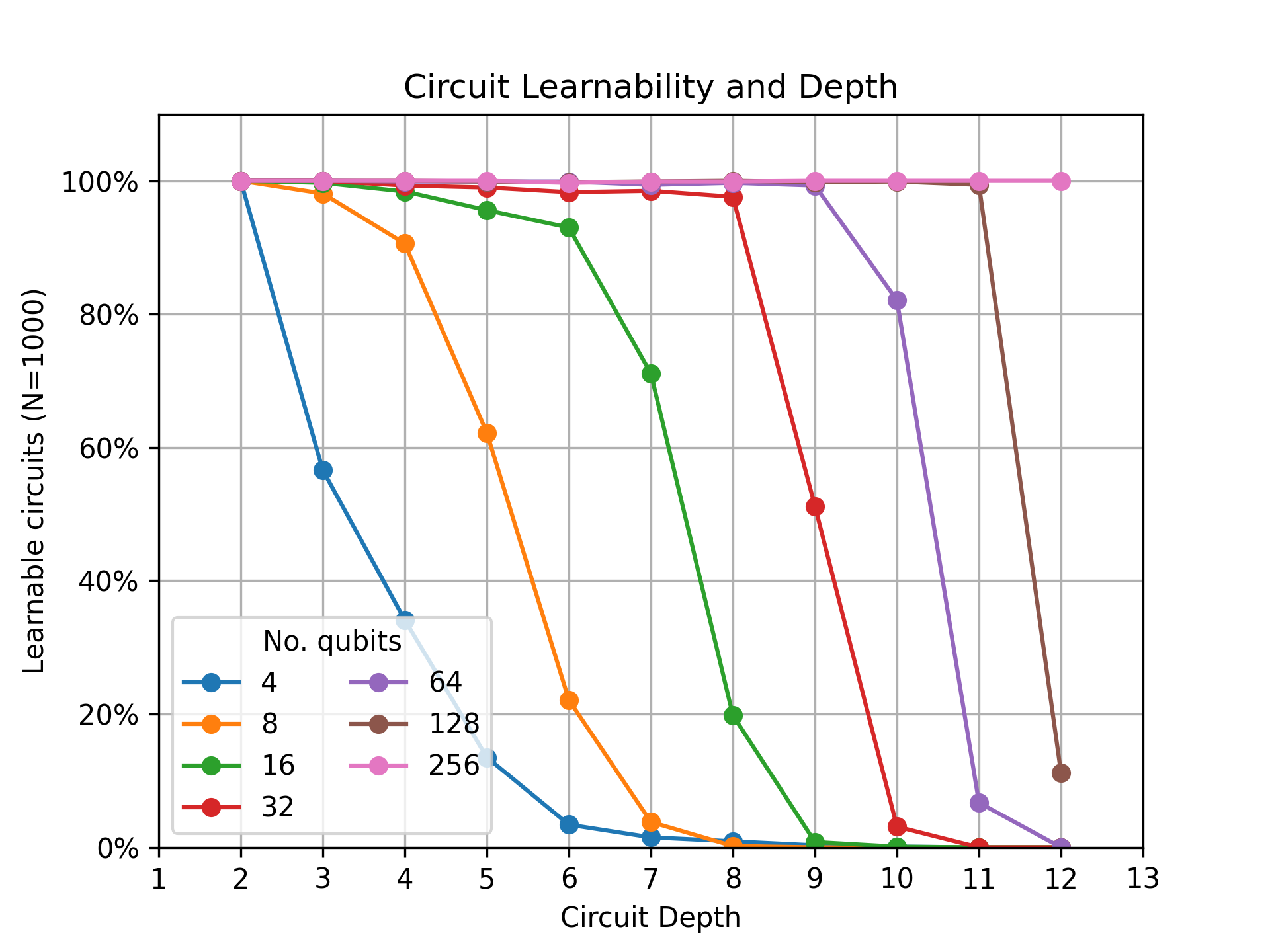}
        \subcaption{}
        \label{subfig:full_learn_b}
    \end{subfigure}
    \begin{subfigure}[]{0.45\textwidth}
        \centering
        \includegraphics[width=\textwidth]{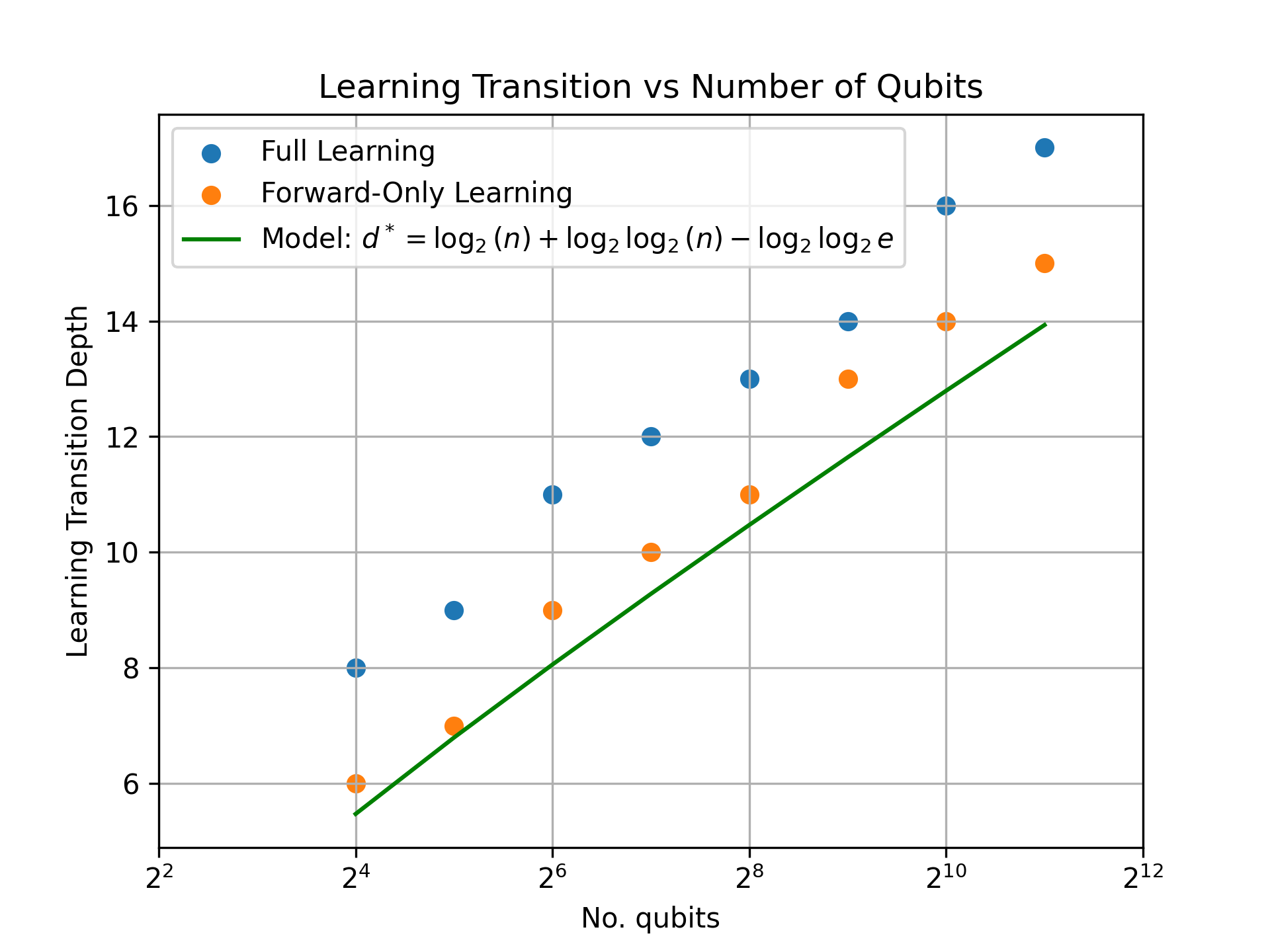}
        \subcaption{}
        \label{subfig:full_learn_c}
    \end{subfigure}
    \caption{Pivot-removal simulations for independent, uniformly random two-qubit
    layers, with consecutive repeated gates grouped before removal. (a) Success probability
    for alternating forward--backward removal at various depths and qubit counts.
    (b) First depth at which empirical success falls below $1/2$, for alternating removal (blue)
    and forward-only removal (orange). The green curve is a heuristic reference for the scale
    $\log_2 n+\log_2\ln n$, not a threshold guarantee for this independent-layer ensemble.}
    \label{fig:full_learning_statistics}
\end{figure}
To compare the two procedures, \Cref{subfig:full_learn_c} plots the first depth at which the
empirical success probability falls below $1/2$. Over the sizes studied, both transition depths
grow approximately linearly with $\log_2 n$. Alternating forward and backward passes gives an
apparently constant additive depth advantage over forward-only removal. These finite-size
observations suggest that the logarithmic transition persists beyond the single-overlap ensemble
of~\Cref{def:random_layered_architecture}, while allowing additional removals by switching
direction. They do not determine an asymptotic threshold or additive offset for independent
layers.
\section{Discussion} \label{sec:discussion}

This work develops a framework for proper learning of shallow unitary circuits, extending local inversion methods beyond geometrically local architectures. We identify
causal boundaries as witnesses for local factorization. For random all-to-all
layered architectures with fixed gate locality $k\ge2$ and no repeated qubit pairs in consecutive
layers, we prove a lightcone transition within a constant-width window around
$\log_k n+\log_k\ln n$. 
For two-qubit gates, \Cref{cor:single_overlap_learning} gives exact proper learning below this
window with polynomial query complexity for a specified well-spaced Haar discretization,
under the supplied-gate assumptions stated there.

An immediate direction is to extend this learning guarantee to other gate families throughout
the depth range below the lightcone-threshold window. For $k>2$, this requires identifying residual factor
supports and justifying the tests used for further factorization. For continuous gate families,
approximate inversions additionally require control of accumulated errors while preserving the
original layout and gate-family constraints. Our numerical results also motivate
extending the threshold analysis to independently sampled layers, with repeated
interactions treated appropriately, and quantifying the depth advantage of alternating forward
and backward removal.

Some of the framework may extend beyond quantum circuits to other computational circuit models.
Our lightcone analysis for random all-to-all $k$-local architectures depends only on the assigned
gate supports, not on the operations performed by the gates. The resulting structural threshold
therefore applies to the same connectivity model in other settings. Its consequences for learning
depend additionally on whether directed reachability produces detectable influence, a property
that must be established for the chosen operations.

Finally, it remains an open question to what extent these local-inversion techniques
can be applied beyond the query setting. Copies of a fixed output state or classical shadows do not
generally allow queries on input states of the learner's choosing, so our guarantees do not directly
apply to these access models. Other interesting quantum circuit learning problems
include the learning of state preparation circuits given copies $C \ket{0}^{\otimes n}$ of the
output state, and learning circuits from a given list of measurements or expectation values. In
these cases, probing causal structure becomes more challenging, and it is not clear how to
generalize the present techniques. Understanding proper circuit learning in these alternative
settings will furthermore provide a more realistic view on the security of recent cryptographic
schemes~\cite{fefferman2025hardness, niroula2026digital}, because there the learning task which
breaks the protocols is closer to these alternative frameworks.

\paragraph{Acknowledgments} The authors are grateful for discussions with Shouvanik Chakrabarti,
which have helped us complete one of the proofs, and for his general leadership and support.

\paragraph{AI Statement} An earlier version of this work, posted to arXiv in August 2026, used LLM 
assistance for completing certain combinatorial proofs. In the present version, GPT-6.1 Sol provided 
a proof of the threshold that did not require assumptions made in the initial version, resulting in 
significant modifications to Section 5. Sol was also used to check for correctness and improve writing. 
The authors have verified all LLM outputs and take full responsibility for the correctness of this work. 

\paragraph{Disclaimer} This paper was prepared for informational purposes by the Global Technology
Applied Research center of JPMorgan Chase \& Co. This paper is not a product of the Research
Department of JPMorgan Chase \& Co. or its affiliates. Neither JPMorgan Chase \& Co. nor any of its
affiliates makes any explicit or implied representation or warranty and none of them accept any
liability in connection with this paper, including, without limitation, with respect to the
completeness, accuracy, or reliability of the information contained herein and the potential legal,
compliance, tax, or accounting effects thereof. This document is not intended as investment research
or investment advice, or as a recommendation, offer, or solicitation for the purchase or sale of any
security, financial instrument, financial product or service, or to be used in any way for
evaluating
the merits of participating in any transaction.

\printbibliography

@inproceedings{bakshi2024structure,
  title={Structure learning of Hamiltonians from real-time evolution},
  author={Bakshi, Ainesh and Liu, Allen and Moitra, Ankur and Tang, Ewin},
  booktitle={2024 IEEE 65th Annual Symposium on Foundations of Computer Science (FOCS)},
  pages={1037--1050},
  year={2024},
  organization={IEEE}
}

@article{bairey2019learning,
  title={Learning a local Hamiltonian from local measurements},
  author={Bairey, Eyal and Arad, Itai and Lindner, Netanel H},
  journal={Physical review letters},
  volume={122},
  number={2},
  pages={020504},
  year={2019},
  publisher={APS}
}

@article{brouwer2015lossy,
  title={Lossy Gossip and Composition of Metrics},
  author={Brouwer, Andries E. and Draisma, Jan and Frenk, Bart J.},
  journal={Discrete \& Computational Geometry},
  volume={53},
  number={4},
  pages={890--913},
  year={2015},
  doi={10.1007/s00454-015-9666-1}
}

@article{casteigts2024sharp,
  title={Sharp Thresholds in Random Simple Temporal Graphs},
  author={Casteigts, Arnaud and Raskin, Michael and Renken, Malte and Zamaraev, Viktor},
  journal={SIAM Journal on Computing},
  volume={53},
  number={2},
  pages={346--388},
  year={2024},
  doi={10.1137/22M1511916}
}

@inproceedings{doerr2017randomized,
  title={Randomized Rumor Spreading Revisited},
  author={Doerr, Benjamin and Kostrygin, Anatolii},
  booktitle={44th International Colloquium on Automata, Languages, and Programming (ICALP 2017)},
  series={Leibniz International Proceedings in Informatics},
  volume={80},
  pages={138:1--138:14},
  year={2017},
  publisher={Schloss Dagstuhl--Leibniz-Zentrum f{\"u}r Informatik},
  doi={10.4230/LIPIcs.ICALP.2017.138}
}

@article{doerr2023randomized,
  title={Randomized Rumor Spreading Revisited (Long Version)},
  author={Doerr, Benjamin and Kostrygin, Anatolii},
  journal={arXiv preprint arXiv:2303.11150},
  year={2023},
  url={https://arxiv.org/abs/2303.11150v1}
}

@article{fefferman2024anti,
  title={Anti-concentration for the unitary Haar measure and applications to random quantum circuits},
  author={Fefferman, Bill and Ghosh, Soumik and Zhan, Wei},
  journal={arXiv preprint arXiv:2407.19561},
  year={2024}
}

@article{fefferman2025hardness,
  title={The Hardness of Learning Quantum Circuits and its Cryptographic Applications},
  author={Fefferman, Bill and Ghosh, Soumik and Sinha, Makrand and Yuen, Henry},
  journal={arXiv preprint arXiv:2504.15343},
  year={2025}
}

@article{fenner2018gossip,
  title={{NP}-completeness in the gossip monoid},
  author={Fenner, Peter and Johnson, Marianne and Kambites, Mark},
  journal={International Journal of Algebra and Computation},
  volume={28},
  number={4},
  pages={653--672},
  year={2018},
  doi={10.1142/S0218196718500297}
}

@article{mckay2003asymptotic,
  title={Asymptotic enumeration of 0--1 matrices with equal row sums and equal column sums},
  author={McKay, Brendan D and Wang, Xiaoji},
  journal={Linear algebra and its applications},
  volume={373},
  pages={273--287},
  year={2003},
  publisher={Elsevier}
}

@article{gu2024practical,
  title={Practical Hamiltonian learning with unitary dynamics and Gibbs states},
  author={Gu, Andi and Cincio, Lukasz and Coles, Patrick J},
  journal={Nature Communications},
  volume={15},
  number={1},
  pages={312},
  year={2024},
  publisher={Nature Publishing Group UK London}
}

@inproceedings{huang2024learning,
  title={Learning shallow quantum circuits},
  author={Huang, Hsin-Yuan and Liu, Yunchao and Broughton, Michael and Kim, Isaac and Anshu, Anurag and Landau, Zeph and McClean, Jarrod R},
  booktitle={Proceedings of the 56th Annual ACM Symposium on Theory of Computing},
  pages={1343--1351},
  year={2024}
}

@inproceedings{karp2000randomized,
  title={Randomized rumor spreading},
  author={Karp, Richard M. and Schindelhauer, Christian and Shenker, Scott and V{\"o}cking, Berthold},
  booktitle={Proceedings 41st Annual Symposium on Foundations of Computer Science},
  pages={565--574},
  year={2000},
  organization={IEEE},
  doi={10.1109/SFCS.2000.892324}
}

@article{kim2024learning,
  title={Learning state preparation circuits for quantum phases of matter},
  author={Kim, Hyun-Soo and Kim, Isaac H and Ranard, Daniel},
  journal={arXiv preprint arXiv:2410.23544},
  year={2024}
}

@inproceedings{landau2025learning,
  title={Learning quantum states prepared by shallow circuits in polynomial time},
  author={Landau, Zeph and Liu, Yunchao},
  booktitle={Proceedings of the 57th Annual ACM Symposium on Theory of Computing},
  pages={1828--1838},
  year={2025}
}

@book{nielsen2010quantum,
  title={Quantum computation and quantum information},
  author={Nielsen, Michael A and Chuang, Isaac L},
  year={2010},
  publisher={Cambridge university press}
}

@article{niroula2026digital,
  title={Digital signatures with classical shadows on near-term quantum computers},
  author={Niroula, Pradeep and Liu, Minzhao and Omanakuttan, Sivaprasad and Amaro, David and Chakrabarti, Shouvanik and Ghosh, Soumik and He, Zichang and Jin, Yuwei and Kaleoglu, Fatih and Kordonowy, Steven and others},
  journal={arXiv preprint arXiv:2602.04859},
  year={2026}
}

\newpage

\begin{appendices}
\section{Causal cuts and equivalent causal boundaries} \label{app:causal_cuts}

The gate-removal witnesses of~\Cref{lem:causal_boundary_iff_gate_remove} admit a second description.
Moving the root of a forward lightcone past a front gate loses the same ordered qubit pairs
that backward lightcone growth gains across that gate. Causal cuts make this equivalence precise
without assuming a layered architecture or padding idle wires with identity gates.

A set $S$ of gates is \emph{predecessor-closed} if, whenever it contains a gate $G$, it also
contains every gate from which a directed path leads to $G$. Its associated \emph{causal cut}
consists of one wire per qubit: the wire immediately after the last gate on that qubit
belonging to $S$, or the input wire if no such gate exists.
The gates in $S$ can be executed before any gate outside $S$, so the cut marks the boundary
between an executed part of the circuit and the remaining gates. It need not be a layer boundary
in a given scheduling.

Choose nested predecessor-closed gate sets $S_0\subseteq\cdots\subseteq S_m$, with $S_0$ empty
and $S_m$ equal to the full gate set of $C$, and let $\Sigma_i$ be the cut associated with $S_i$.
One possible choice is to add gates one at a time in an order that respects all causal
dependencies, so consecutive sets differ by exactly one gate. For layered circuits, one can
instead choose the layer boundaries. Write $q(\Sigma_i)$ for the wire carrying qubit $q$ at $\Sigma_i$.
For $0\le i\le j\le m$, define the following lightcones, allowing zero-gate paths:
\begin{equation}
    \begin{aligned}
        F_{i,j}(q)&\coloneqq
        \{q'\in Q:q(\Sigma_i)\text{ reaches }q'(\Sigma_j)\},\\
        B_{i,j}(q')&\coloneqq
        \{q\in Q:q(\Sigma_i)\text{ reaches }q'(\Sigma_j)\}.
    \end{aligned}
\end{equation}
The indices specify two endpoint cuts: $F_{i,j}$ is rooted at $\Sigma_i$, while $B_{i,j}$ is
rooted at $\Sigma_j$. The forward and backward lightcones describe the same directed
reachability relation, so
\begin{equation} \label{eq:cut_reachability_duality}
    q'\in F_{i,j}(q)\iff q\in B_{i,j}(q').
\end{equation}

The \emph{changing-root} boundaries generalize the causal boundaries
of~\Cref{def:causal_boundary} by moving the root between causal cuts while fixing the opposite
endpoint. A \emph{growth-oriented} boundary instead fixes the root and extends the opposite
endpoint. The lightcones are nested because paths can be extended along the same qubit wires
through intervening gates.

\begin{proposition}[Equivalent cut boundaries] \label{prop:cut_boundary_equivalence}
    For any circuit architecture, nested causal cuts as above, $0\le i<j\le m$, and
    $q,q'\in Q$,
    \begin{equation}
        \begin{aligned}
            q'\in F_{i,j}(q)\setminus F_{i,j-1}(q)
            &\iff q\in B_{i,j}(q')\setminus B_{i,j-1}(q'),\\
            q\in B_{i,j}(q')\setminus B_{i+1,j}(q')
            &\iff q'\in F_{i,j}(q)\setminus F_{i+1,j}(q).
        \end{aligned}
    \end{equation}
\end{proposition}
\begin{proof}
    Apply~\Cref{eq:cut_reachability_duality} to $(i,j)$ and $(i,j-1)$, requiring the first
    membership and excluding the second, to obtain the first assertion. Applying it to
    $(i,j)$ and $(i+1,j)$ gives the second.
\end{proof}

\noindent For the single-gate interpretation, suppose that consecutive cuts differ by exactly one gate,
$S_{i+1}\setminus S_i=\{G\}$. Then $G$ is a front gate of the subcircuit between $\Sigma_i$
and $\Sigma_j$. Replacing $G$ by wires leaves the reachability described by $F_{i+1,j}$.
Thus, $F_{i,j}(q)\setminus F_{i+1,j}(q)$ gives exactly the disconnection witnesses
of~\Cref{lem:causal_boundary_iff_gate_remove}, and the second equivalence describes them by
backward growth. The first equivalence treats a single back gate in the same way.
For learning, the cuts bound the current residual circuit; this description requires no
additional oracle access to internal wires and does not replace the signal-propagation conditions.

In a layered circuit, choosing cuts at layer boundaries recovers the gate-distance lightcones
of the main text. The proposition then relates growth with depth in one direction to
changing-root boundaries in the other.

\begin{corollary}[Layered circuits] \label{cor:layered_cut_boundaries}
    Let $C$ be a layered circuit with $d\ge1$ layers, and choose $\Sigma_t$ immediately after layer $t$, with
    $\Sigma_0$ the input cut. Then, for $0\le i\le j\le d$,
    \begin{equation}
        F_{i,j}(q)=\vec{L}_{i,j-i}(q),
        \qquad
        B_{i,j}(q')=\backvec{L}_{j,j-i}(q').
    \end{equation}
    Consequently,
    \begin{equation}
        \begin{aligned}
            q'\in\vec{L}_{0,d}(q)\setminus\vec{L}_{0,d-1}(q)
            &\iff q\in\partial\backvec{L}_d(q'),\\
            q\in\backvec{L}_{d,d}(q')\setminus\backvec{L}_{d,d-1}(q')
            &\iff q'\in\partial\vec{L}_1(q).
        \end{aligned}
    \end{equation}
\end{corollary}
\begin{proof}
    Every directed path between cuts $\Sigma_i$ and $\Sigma_j$ traverses exactly $j-i$ gates.
    Any qubit label reached earlier can be followed along its own wire to the terminal cut.
    Thus, the qubit sets reached within $j-i$ gates are exactly the cut lightcones, in either
    direction. Apply~\Cref{prop:cut_boundary_equivalence} with $i=0$ and $j=d$ to obtain the
    boundary equivalences.
\end{proof}

These equivalences relate growth in one direction to changing-root boundaries in the other.
The forward and backward causal boundaries of the main text can still identify different
pivots. In particular, \Cref{prop:front_back_front} exhibits a circuit with a back pivot but
no front pivots, where removing the back pivot creates front pivots. Alternating front and
back removal can therefore clear circuits that neither direction alone can clear, as
illustrated in~\Cref{fig:forward-back-forward-example}.

For general architectures, cut growth is also distinct from increasing a gate-count radius:
paths between cuts need not traverse the same number of gates, so the proposition does not
assert the corresponding equivalences for gate-distance truncation.

\section{Proofs for the random-architecture analysis} \label{app:delayed_proofs}

This appendix supplies the proofs omitted from the random-architecture analysis
in~\Cref{sec:random_circuit_ensemble}, in order of appearance.
        
\begin{proof}[Proof of~\Cref{lem:acceptance}]
    Relabeling the qubits shows that $p_{n,k}$ does not depend on $P$.
    Write $m=n/k\ge k$. To show that $p_{n,k}>0$, label the qubits in the blocks
    of $P$ by $(i,a)$, where $0\le i<m$ and $0\le a<k$. The blocks
    \begin{equation}
        \{(j+a\bmod m,\;a):0\le a<k\},\qquad 0\le j<m,
    \end{equation}
    form a partition compatible with $P$: since $m\ge k$, each new block meets
    $k$ distinct old blocks, and every qubit occurs exactly once.
    Thus at least one of the
    $\cal{N}_{n,k}\coloneqq n!/(m!(k!)^m)$ partitions is compatible with $P$.
    
    To derive the limit, label $m$ prospective blocks of $Q$, each with $k$ positions, and assign
    the $n$ qubits uniformly to these positions. Each partition $Q$ arises
    from exactly $m!(k!)^m$ assignments, so this procedure samples $Q$ uniformly.
    Viewing each qubit as an edge between its block in $P$ and its assigned
    block in $Q$ gives the $k$-regular bipartite configuration model.
    The number of edges between two blocks is their intersection size,
    so compatibility is equivalent to the graph having no parallel edges.
    
    There are $n!$ equally likely assignments. Each simple labeled
    $k$-regular bipartite graph is realized by $(k!)^{2m}$ assignments.
    Thus, if $B_{m,k}$ counts these graphs,
    \begin{equation}
        p_{n,k}=\frac{B_{m,k}(k!)^{2m}}{n!}.
    \end{equation}
    The enumeration formula of McKay and Wang~\cite{mckay2003asymptotic},
    specialized to $m\times m$ zero--one matrices with every row and column
    sum equal to the fixed value $k$, gives
    \begin{equation}
        B_{m,k}=\frac{n!}{(k!)^{2m}}
        \exp\left(-\frac{(k-1)^2}{2}+O_k(n^{-1})\right),
    \end{equation}
    which proves the claimed limit.

    For fixed $k$, the positive limit bounds $p_{n,k}$ away from zero for all
    sufficiently large $n$. Together with $p_{n,k} > 0$, this implies
    a strictly positive lower bound $p_*(k)$. 
\end{proof}

\subsection{Early-depth estimates} \label{app:early_depth}

Here we prove~\Cref{lem:early_growth}. We first bound collision moments for a uniform
partition, then account for compatibility with the preceding layer.
We apply the resulting conditional moment bound to the accumulated collision losses
along the evolving lightcones.

\begin{lemma}[Collision moment bound] \label{lem:collision_mgf}
    For a fixed set $V$ of size $s<n$ and an independent uniform partition $Q$ into blocks of
    size $k$,
    \begin{equation} \label{eq:collision_mgf}
        \ln \mathbb{E}e^{uD_Q(V)}
        \le \frac{(k-1)s^2}{n-s}(e^u-1),
        \qquad u\ge0.
    \end{equation}
\end{lemma}
\begin{proof}
    Generate $Q$ by grouping a uniform random permutation into blocks of $k$ positions. Expose
    the positions occupied by the $s$ vertices of $V$ one at a time. Let $I_i$ indicate that
    the $i$th vertex enters a block already occupied by an earlier vertex. Then
    $D_Q(V)=\sum_{i=1}^s I_i$. Before the $i$th exposure, at most $(k-1)(i-1)$ unexposed
    positions lie in previously occupied blocks, so
    \begin{equation}
        \Pr(I_i=1\mid\text{earlier positions})
        \le\frac{(k-1)(i-1)}{n-i+1}
        \le\frac{(k-1)s}{n-s}.
    \end{equation}
    It follows that
    \begin{equation}
        \mathbb{E}[e^{uI_i}\mid\text{earlier positions}]
        \le1+\frac{(k-1)s}{n-s}(e^u-1).
    \end{equation}
    Iterating conditional expectations over the $s$ exposures gives
    \begin{equation}
        \mathbb{E}e^{uD_Q(V)}
        \le\left[1+\frac{(k-1)s}{n-s}(e^u-1)\right]^s
        \le\exp\left[\frac{(k-1)s^2}{n-s}(e^u-1)\right],
    \end{equation}
    where the last inequality uses $1+x\le e^x$.
\end{proof}

To condition on compatibility, we use the following elementary inequality.
\begin{lemma}[Conditional exponential moments] \label{lem:conditional_mgf}
    Let $E$ be an event of probability $p>0$, and let $Z\ge0$.
    For $u\ge0$ and $t\ge1$, if $\mathbb{E}e^{uZ}\le t$, then
    \begin{equation} \label{eq:conditional_mgf}
        \mathbb{E}[e^{uZ}\mid E]\le t^{1/p}.
    \end{equation}
\end{lemma}
\begin{proof}
    Since $e^{uZ}-1\ge0$,
    \begin{equation}
        \mathbb{E}[e^{uZ}\mid E]
        \le1+\frac{\mathbb{E}(e^{uZ}-1)}p
        \le1+\frac{t-1}p
        \le t^{1/p},
    \end{equation}
    where the last step is an application of Bernoulli's inequality.
\end{proof}

Combining~\Cref{lem:collision_mgf,lem:conditional_mgf} gives the collision moment bound
for the single-overlap layered ensemble.
\begin{lemma}[Conditional collision moment bound] \label{lem:conditional_collision_mgf}
    Let $\cal{F}_j$ record the full circuit history through the $j$th layer after gate $G$ in the
    single-overlap layered ensemble. For any $\cal{F}_j$-measurable set $V_j$
    with $|V_j|<n$ and every $u\ge0$,
    \begin{equation} \label{eq:conditional_collision_mgf}
        \ln \mathbb{E}\left[e^{uD_j(V_j)}\mid\cal{F}_j\right]
        \le \frac{(k-1)|V_j|^2}{p_{n,k}(n-|V_j|)}(e^u-1).
    \end{equation}
\end{lemma}
\begin{proof}
    Fix this history, which determines $V_j$ and the preceding partition $P=P_{\ell+j}$.
    The next circuit layer has the law of an independent uniform partition $Q$
    conditioned on compatibility $E=\{P\sim Q\}$.
    Consequently,
    \begin{equation} \label{eq:Dj_to_conditionE}
        \mathbb{E}\left[e^{uD_j(V_j)}\mid\cal{F}_j\right]
        =
        \mathbb{E}\left[e^{uD_Q(V_j)}\mid\cal{F}_j,E\right].
    \end{equation}
    For $u \geq 0$, \Cref{lem:collision_mgf} shows that
    $\mathbb{E}[e^{u D_Q(V_j)}\mid\cal{F}_j] \leq t_j$, with
    \begin{equation}
        t_j\coloneqq\exp\left[
            \frac{(k-1) \abs{V_j}^2}{n-\abs{V_j}}(e^u - 1)
        \right].
    \end{equation}
    $\Pr(E\mid\cal{F}_j)=p_{n,k}$ by~\Cref{lem:acceptance}. Plugging the right side of \Cref{eq:Dj_to_conditionE} into~\Cref{lem:conditional_mgf}, with parameters $t = t_j$, $Z = D_Q(V_j)$, and $E = \{P\sim Q\}$, gives
    \begin{equation}
        \mathbb{E}\left[e^{uD_j(V_j)}\mid\cal{F}_j\right]
        \leq \exp\left[
            \frac{1}{p_{n,k}}\frac{(k-1) \abs{V_j}^2}{n- \abs{V_j}}(e^u - 1)
        \right].
    \end{equation}
\end{proof}

\begin{proof}[Proof of~\Cref{lem:early_growth}]
    Take $V_j$ to be $U_j$ or $A_j$.
    At every step, $\abs{A_j}\le\abs{U_j}\le2k^j$.
    Choose $K\ge\log_k4$; then, for $1\le j < J$,  
    \begin{equation}
        |V_j|\le2k^j\le2n/k^K\le n/2.
    \end{equation}
    Hence
    \begin{equation}
        \frac{1}{p_*(k)}\frac{(k-1)|V_j|^2}{n-|V_j|}
        \le \frac{8(k-1)}{p_*(k)}\frac{k^{2j}}n.
    \end{equation}
    Applying~\Cref{lem:conditional_collision_mgf}, for either choice and every $u\ge0$,
    \begin{equation} \label{eq:early_collision_intensity}
        \ln \mathbb{E}\left[e^{uD_j(V_j)}\mid\cal{F}_j\right]
        \le \lambda_k\frac{k^{2j}}n(e^u-1),
    \end{equation}
    where $\lambda_k=8(k-1)/p_*(k)$ is independent of the history.

    Define the weighted collision sums
    \begin{equation}
        \Psi_V\coloneqq\sum_{j=1}^{J-1}k^{-j}D_j(V_j),
        \qquad V=A,U.
    \end{equation}
    After the first subsequent layer, $S_1=k$ and $H_1=|U_1|=2k$,
    by single-qubit overlap.
    Dividing~\Cref{eq:early_growth_bounds} by $k^{j+1}$ gives
    \begin{equation}
            \frac{S_{j+1}}{k^{j+1}}
            =\frac{S_j}{k^j}-k^{-j}D_j(A_j),\qquad
            \frac{H_{j+1}}{k^{j+1}}
            \ge\frac{H_j}{k^j}-2k^{-j}D_j(U_j).
    \end{equation}
    Summing from $j=1$ to $J-1$ and using the initial values yields
    \begin{equation} \label{eq:normalized_early_growth}
        \frac{S_J}{k^J}=1-\Psi_A,
        \qquad
        \frac{H_J}{k^J}\ge2-2\Psi_U.
    \end{equation}
    It remains to bound either weighted sum, denoted below by $\Psi$.

    For brevity, write $L=\ln n$, and choose the exponential-moment parameter $\theta=kL$.
    Applying~\Cref{eq:early_collision_intensity} with $u=\theta k^{-j}$ and
    iterating conditional expectations yields
    \begin{equation} \label{eq:weighted_collision_mgf}
        \ln\mathbb{E}[e^{\theta\Psi}\mid\cal{F}_1]
        \le\frac{\lambda_k}{n}\sum_{j=1}^{J-1}
        k^{2j}\left(e^{L/k^{j-1}}-1\right).
    \end{equation}
    We bound the right-hand side in three ranges of $j$.
    First, the contribution from $j=1$ is $O_k(1)$.

    Next, consider indices satisfying both $j\ge2$ and $k^{j-1}<L$.
    For these terms, $e^{L/k^{j-1}}\le n^{1/k}$ and $k^{2j}\le k^2L^2$.
    There are $O_k(\ln L)$ such indices, so their total contribution is
    \begin{equation}
        O_k\left(n^{-1+1/k}L^2\ln L\right)=o(1).
    \end{equation}
    Finally, consider indices satisfying $j\ge2$ and $k^{j-1}\ge L$.
    Here $L/k^{j-1}\le1$. Applying $e^u-1\le2u$ for $0\le u\le1$ and
    enlarging the sum to all indices gives a contribution bounded by
    \begin{equation}
        O_k\left(\frac{L}{n}\sum_{j=1}^{J-1}k^{j+1}\right)
        =O_k(k^{-K}L).
    \end{equation}
    Consequently, there are constants $\alpha_k,\beta_k$ independent of $K$ once $K$ is sufficiently
    large, such that
    \begin{equation}
        \ln\mathbb{E}[e^{\theta\Psi}\mid\cal{F}_1]
        \le \alpha_k+\beta_k k^{-K}\ln n+o(1).
    \end{equation}

    By~\Cref{eq:normalized_early_growth}, the event $S_J<\delta k^J$ is equivalent to
    $\Psi_A>1-\delta$, while $H_J<\delta k^J$ implies $\Psi_U>1-\delta/2$.
    Applying exponential Markov bounds to these weighted collision sums gives,
    for some constant $C>0$ depending only on $k$,
    \begin{equation}
        \begin{aligned}
            \Pr(S_J<\delta k^J\mid\cal{F}_1)
            &\le C n^{-k+k\delta+\beta_k k^{-K}},\\
            \Pr(H_J<\delta k^J\mid\cal{F}_1)
            &\le C n^{-k+k\delta/2+\beta_k k^{-K}}.
        \end{aligned}
    \end{equation}
    Choose $0<\delta<1$ small enough and then $K$ large enough that
    $k\delta+\beta_k k^{-K}<\varepsilon$. Both bounds are then at most $Cn^{-k+\varepsilon}$, uniformly over
    $\cal{F}_1$. Since $k^J\ge nk^{-K-1}$, taking $c=\delta k^{-K-1}$ gives
    $cn\le\delta k^J$. Averaging over the first subsequent layer therefore yields
    the stated bounds for $\sigma_J$ and $\eta_J$.
\end{proof}

\subsection{Late-depth estimates} \label{app:late_depth}

Here we prove~\Cref{lem:late_persistence,lem:late_saturation}, which provide estimates
for depths below and above the saturation threshold, respectively, beyond the
early-growth regime studied in~\Cref{app:early_depth}.
Both estimates use the following count of blocks contained in a fixed set.

\begin{lemma}[Blocks contained in a fixed set] \label{lem:contained_block_moments}
    Let $k\mid n$ and let $Q$ be a uniformly random partition of $n$ qubits into blocks
    of size $k$. For a fixed set $S$ of size $s$, let $N_S$ count the blocks of $Q$
    contained entirely in $S$. For integers $0\le q\le n/k$,
    \begin{equation} \label{eq:contained_block_factorial_moments}
        \mathbb{E}(N_S)_q
        =\left(\frac nk\right)_q\frac{(s)_{kq}}{(n)_{kq}}
        \le\left(\frac{s^k}{kn^{k-1}}\right)^q.
    \end{equation}
    In particular, the case $q=1$ gives
    \begin{equation} \label{eq:contained_block_mean}
        k\mathbb{E}N_S=n\frac{(s)_k}{(n)_k}
        \le\frac{s^k}{n^{k-1}}.
    \end{equation}
\end{lemma}
\begin{proof}
    Generate $Q$ by grouping a uniform random permutation into $n/k$ labeled blocks of size $k$.
    The falling factorial $(N_S)_q$ counts ordered choices of $q$ distinct blocks contained in $S$.
    There are $(n/k)_q$ ordered choices of $q$ distinct blocks in total.
    Each contributes to $(N_S)_q$ exactly when all its $kq$ qubits lie in $S$,
    an event of probability $(s)_{kq}/(n)_{kq}$. Hence
    \begin{equation}
        \begin{aligned}
            \mathbb{E}(N_S)_q
            &=\left(\frac nk\right)_q\frac{(s)_{kq}}{(n)_{kq}}\\
            &\le\left(\frac nk\right)^q\left(\frac sn\right)^{kq}
            =\left(\frac{s^k}{kn^{k-1}}\right)^q.
        \end{aligned}
    \end{equation}
    When $kq>s$, the numerator is zero; moments with $q>n/k$ are also zero.
    Taking $q=1$ gives~\Cref{eq:contained_block_mean}.
\end{proof}

\subsubsection{Below the saturation threshold} \label{app:pre_saturation}

To prove~\Cref{lem:late_persistence}, we represent a uniform partition as a uniform permutation
of qubits grouped into consecutive sets of $k$ positions, and use the following concentration bound.
\begin{lemma}[Permutation concentration] \label{lem:permutation_concentration}
    Let $f$ be a real-valued function on permutations of $n$ elements.
    Suppose that transposing any two positions in any permutation changes $f$ by at most $L>0$.
    For a uniformly random permutation $\pi$, write $F=f(\pi)$. Then, for $z\geq0$,
    \begin{equation} \label{eq:permutation_concentration}
        \Pr(F-\mathbb{E}F\ge z)\le e^{-z^2/(2nL^2)},
        \qquad
        \Pr(F-\mathbb{E}F\le-z)\le e^{-z^2/(2nL^2)}.
    \end{equation}
\end{lemma}
\begin{proof}
    Expose the permutation one position at a time and let
    \begin{equation}
        M_i\coloneqq\mathbb{E}[F\mid\pi(1),\ldots,\pi(i)],
        \qquad 0\le i\le n
    \end{equation}
    be the Doob martingale of $F$. Its endpoints are $M_0=\mathbb{E}F$ and $M_n=F$.

    We claim that the martingale increments satisfy
    \begin{equation}
        |M_i-M_{i-1}|\le L,\qquad 1\le i\le n.
    \end{equation}
    Fix the previously exposed positions and two possible next values $a,b$.
    Pair each completion with $a$ next with the completion obtained by swapping $a$
    with the later occurrence of $b$. This is a bijection between the two sets of
    completions, leaving the exposed positions unchanged.
    Paired permutations differ by one transposition, so their values of $f$ differ
    by at most $L$. Since the completions are conditionally uniform, the corresponding
    conditional expectations also differ by at most $L$.
    Since $M_{i-1}$ averages these possible values of $M_i$, the claimed increment bound follows.

    Azuma's inequality for martingales with bounded increments, applied to
    $M_n-M_0=F-\mathbb{E}F$, gives both tails.
\end{proof}

\begin{proof}[Proof of~\Cref{lem:late_persistence}]
    Condition on $\cal{F}_j$, and write $X=A_j\setminus B_j$,
    $Y=B_j\setminus A_j$, $x=|X|$, and $y=|Y|$, so $x+y=h=H_j$.
    First sample an independent uniform reference partition. Let $N_X,N_Y$ count its gates
    lying entirely in $X,Y$, respectively. Define $F\coloneqq k(N_X+N_Y)$, the number
    of qubits in these gates. Then $F\le H_{j+1}$, and applying~\Cref{lem:contained_block_moments}
    to $X$ and $Y$ gives
    \begin{equation}
        \mathbb{E}F=n\frac{(x)_k+(y)_k}{(n)_k}
        \ge4^{-k}\frac{h^k}{n^{k-1}}.
    \end{equation}
    Indeed, $\max(x,y)\ge h/2\ge2k$, so the larger falling factorial is at least $(h/4)^k$,
    while $(n)_k\le n^k$.

    A transposition affects at most two gates, and hence changes $F$ by at most $2k$.
    By~\Cref{lem:permutation_concentration},
    \begin{equation}
        \Pr(F<\mathbb{E}F/2)\le
        \exp\left(-\frac{(\mathbb{E}F)^2}{32nk^2}\right)
        \le \exp\left(-a_k\frac{h^{2k}}{n^{2k-1}}\right)
    \end{equation}
    for some $a_k>0$. Taking $b_k=4^{-k}/2$ bounds the desired failure event under the
    reference partition. Conditioning on compatibility with the preceding partition
    increases its probability by at most $p_{n,k}^{-1}\le p_*(k)^{-1}$.
    Substituting $h=n\eta_j$ and $H_{j+1}=n\eta_{j+1}$ gives~\Cref{eq:late_persistence}.
\end{proof}

\subsubsection{Above the saturation threshold} \label{app:post_saturation}

We first derive~\Cref{eq:late_complement_tail} from the complement's exponential moments,
then prove~\Cref{lem:late_saturation}.

\begin{lemma}[Complement exponential moments] \label{lem:complement_moments}
    Fix a set of $r$ unreached qubits and a preceding partition $P$.
    Let $Q$ be sampled uniformly among partitions compatible with $P$, and let $R'$
    be the number of qubits still unreached after this layer. Then
    \begin{equation} \label{eq:complement_mgf}
        \ln \mathbb{E}e^{uR'/k}
        \le \frac1{p_{n,k}}\frac{r^k}{kn^{k-1}}(e^u-1),
        \qquad u\ge0.
    \end{equation}
    In particular, \Cref{eq:late_complement_tail} holds.
\end{lemma}
\begin{proof}
    First sample an independent uniform reference partition $Q$.
    For the current complement $S$ of size $r$, $N_S=R'/k$ is the number of
    unreached blocks. Using~\Cref{lem:contained_block_moments} and expanding the moment
    generating function in factorial moments gives
    \begin{equation}
        \mathbb{E}e^{uN_S}
        =\sum_{q=0}^\infty \frac{\mathbb{E}(N_S)_q}{q!}(e^u-1)^q
        \le\exp\left[\frac{r^k}{k n^{k-1}}(e^u-1)\right].
    \end{equation}
    Conditioning on compatibility with $P$, an event of probability $p_{n,k}$,
    gives~\Cref{eq:complement_mgf} by~\Cref{lem:conditional_mgf}.
    Exponential Markov with $u=1$ gives
    \begin{equation}
        \Pr\left(
            R'>\frac{e-1}{p_*(k)}\frac{r^k}{n^{k-1}}+3k\ln n
        \right)\le n^{-3}.
    \end{equation}
    Applying this bound conditional on the current history gives~\Cref{eq:late_complement_tail}.
\end{proof}

\begin{proof}[Proof of~\Cref{lem:late_saturation}]
    We first bring the complement below a sufficiently small fixed fraction of $n$.
    Put $a=(e-1)/p_*(k)$, $A=a^{1/(k-1)}$, and $\rho_*=1/(8A)$.
    Write $r=R_j$, $R'=R_{j+1}$, and $R''=R_{j+2}$.
    Under an independent uniform reference partition, $R'$ changes by at most $2k$
    under a transposition. By~\Cref{lem:contained_block_moments,lem:permutation_concentration},
    conditioning on compatibility gives, for fixed $\xi_0>0$,
    \begin{equation}
        \Pr\left(
            \rho_{j+1}>\rho_j^k+\xi_0
            \,\middle|\,\cal{F}_j
        \right)
        \le p_*(k)^{-1}\exp\left(-\frac{\xi_0^2 n}{8k^2}\right).
    \end{equation}
    Choose an integer $L=L(k,c)$ such that $(1-c)^{k^L}\le\rho_*/2$.
    On successful steps, the relative complement satisfies
    $\rho_{i+1}\le\rho_i^k+\xi_0$, starting from $\rho_0\le1-c$.
    Since $u\mapsto u^k$ is $k$-Lipschitz on $[0,1]$, comparison with the unperturbed
    sequence gives
    \begin{equation}
        \rho_L\le(1-c)^{k^L}+\xi_0\frac{k^L-1}{k-1}\le\rho_*,
    \end{equation}
    upon choosing $\xi_0$ sufficiently small. This holds except with probability
    $O_{k,c}(e^{-\alpha_{k,c}n})$, for some $\alpha_{k,c}>0$.

    Reindex from this layer and put $z_i=A\rho_i$ and $\xi_n=3Ak\ln n/n$.
    By~\Cref{eq:late_complement_tail}, each step satisfies
    \begin{equation}
        z_{i+1}\le z_i^k+\xi_n,\qquad z_0\le1/8,
    \end{equation}
    except with conditional probability at most $n^{-3}$. For large $n$, $[0,1/4]$
    is invariant under this recursion. On that interval the derivative of $z^k$ is
    at most $k4^{-(k-1)}\le1/2$. Comparing with $v_i=8^{-k^i}$ therefore gives
    \begin{equation}
        z_i\le8^{-k^i}+2\xi_n.
    \end{equation}
    At $q=\lceil\log_k\ln n\rceil$, the first term is at most $n^{-\ln8}$, and hence
    \begin{equation}
        R_q\le6k\ln n+o(1)\le M_k\ln n
    \end{equation}
    for a fixed $M_k$. The failure probability over these steps is at most $q n^{-3}$.

    It remains to eliminate the logarithmic complement. If $k\ge3$ and $r\le M_k\ln n$,
    the reference mean in~\Cref{lem:contained_block_moments}, conditioning, and Markov's inequality give
    \begin{equation}
        \Pr(R'>0\mid\cal{F}_j)
        \le p_*(k)^{-1}\frac{r^k}{n^{k-1}}
        =O_k\left(\frac{(\ln n)^k}{n^{k-1}}\right)=o(1/n).
    \end{equation}
    Thus one further layer suffices in this case.

    For $k=2$, we use two layers to obtain a sufficiently small failure probability.
    Let $S$ be the current complement of size $r$, and consider two independent
    uniform reference matchings. Then $R'=2N_S$.
    Applying~\Cref{lem:contained_block_moments} with $q=1,2$ gives
    \begin{equation}
        \mathbb{E}N_S=\frac{(r)_2}{2(n-1)},\qquad
        \mathbb{E}(N_S)_2=\frac{(r)_4}{4(n-1)(n-3)}.
    \end{equation}
    Since $(R')_2=4(N_S)_2+2N_S$, conditioning on the first reference matching yields
    \begin{equation} \label{eq:two_step_extinction}
        \mathbb{E}R''
        =\frac{\mathbb{E}(R')_2}{n-1}
        =\frac{(r)_2}{(n-1)^2}
        +\frac{(r)_4}{(n-1)^2(n-3)}.
    \end{equation}
    Conditioning each successive matching on compatibility increases probabilities
    by at most a factor $p_*(2)^{-1}$. Thus the two-layer reference estimate and
    Markov's inequality give
    \begin{equation}
        \begin{aligned}
            \Pr(R''>0\mid\cal{F}_j)
            &\le p_*(2)^{-2}\mathbb{E}_{\mathrm{ref}}R''\\
            &=O\left(\frac{(\ln n)^2}{n^2}+\frac{(\ln n)^4}{n^3}\right)
            =o(1/n).
        \end{aligned}
    \end{equation}
    Combining the initial $L$ steps, the $q$ contraction steps, and these final layers
    gives the stated time bound. The total failure probability is
    $O_{k,c}(e^{-\alpha_{k,c}n})+q n^{-3}+o(1/n)=o(1/n)$.
    All estimates were conditional on the preceding history and uniform in its values.
    A full cone remains full, so using two final layers also covers $k\ge3$.
\end{proof}

\section{Proper learning with discretized Haar gates} \label{app:discretized_haar_learning}

For ease of reference, we collect and adapt the signal-propagation and discretized-Haar
learning arguments of
\citeauthor{fefferman2024anti}~\cite[Sections~5 and~6.3]{fefferman2024anti} to the
architecture-determined retained circuits in~\Cref{thm:learning_forward_general}.
We use a gate set that is well-spaced with respect to the factorization discrepancy and
establish signal propagation for the retained circuits. We first record properties of the
discretization in~\Cref{def:separated_haar_discretization}, then give the
anti-concentration correction and track its contribution to the learning failure probability.

Specializing to two-qubit gates, $d_\otimes$ from~\Cref{eq:factorization_discrepancy} becomes the quotient
pseudometric obtained by identifying gates that differ by a product of single-qubit unitaries.
We use the distribution $\mu_h$ and gate list $\cal{G}_h$
from~\Cref{def:separated_haar_discretization}.
A finite maximal $h$-spaced set $\cal{G}_h$ with respect to $d_\otimes$ ensures that
every two-qubit unitary is within $d_\otimes$-distance $h$ of a listed gate.
Rounding each unitary to the first such gate in a fixed ordering gives a measurable map $f_h$.
Moreover, for absolute constants $C_0,r_0>0$,
\begin{equation} \label{eq:discrete_gate_count}
    |\cal{G}_h|\le C_0h^{-r_0}.
\end{equation}
Indeed, $d_\otimes\le d_\diamond$, so the set is also $h$-spaced in diamond distance;
the bound follows from a packing estimate in the fixed-dimensional unitary group.

\subsection{Discretization correction for Haar anti-concentration}

We recall the discretization correction used near the end of~\cite{fefferman2024anti},
specifically in the proof of Theorem~6.9,
to transfer their Haar signal-propagation argument to the discrete ensemble.

We begin with their local transfer map~\cite[Section~5]{fefferman2024anti}.
Fix the other circuit inputs and the gates preceding a chosen gate on an input-output path.
As the input qubit varies, differences of the two-qubit states entering this gate span
a real subspace $\cal{S}$ of traceless Hermitian operators of dimension at most three.
When $\dim\cal{S}=3$, label the chosen output wire $q_2$ and the other $q_1$.
Define the linear map
$M_G:\cal{S}\to\operatorname{span}_{\bb{R}}\{X,Y,Z\}$ by
\begin{equation}
    M_G(T)=\Tr_{q_1}(GTG^\dagger),
\end{equation}
where the Pauli operators act on $q_2$. Let $Z_{\cal{S}}(G)\coloneqq |\det M_G|$, with the
domain and codomain equipped with the Frobenius inner products.
The Haar estimate underlying~\cite[Lemma~5.1 and Propositions~5.3--5.4]{fefferman2024anti}
gives absolute constants $A_0,\alpha>0$ such that
\begin{equation} \label{eq:haar_transfer_tail}
    \Pr_{H\sim\mathrm{Haar}}\bigl(Z_{\cal{S}}(H)\le x\bigr)\le A_0x^\alpha,
    \qquad x\ge0,
\end{equation}
with the same constants for every three-dimensional subspace $\cal{S}$ arising as above
and for either choice of output wire.
The uniformity allows the subspace to depend on preceding gates, provided it is fixed
before the current independent gate is sampled.

\begin{lemma}[Discretized Haar estimate] \label{lem:discrete_transfer}
    There are absolute constants $A,\alpha>0$, $s>0$, and $K\ge1$, independent of $h$
    and $\cal{S}$, such that for $G\sim\mu_h$,
    \begin{equation} \label{eq:discrete_transfer_bounds}
        \Pr\bigl(Z_{\cal{S}}(G)\le h\bigr)\le Ah^\alpha,\qquad
        \mathbb{E}\bigl[\max\{h,Z_{\cal{S}}(G)\}^{-s}\bigr]\le K.
    \end{equation}
\end{lemma}
\begin{proof}
    Following~\cite[proof of Theorem~6.9]{fefferman2024anti}, we first compare determinants
    before and after rounding. In fixed Frobenius-orthonormal bases, each matrix entry
    has absolute value at most $2$.
    If $d_\diamond(G,G')\le u$, its change is at most $2u$: a Frobenius-normalized
    two-qubit Hermitian input has trace norm at most $2$, and a normalized output basis
    element has operator norm at most $1$.
    Expanding the determinant into six products therefore gives
    \begin{equation}
        |Z_{\cal{S}}(G)-Z_{\cal{S}}(G')|\le144\,d_\diamond(G,G').
    \end{equation}
    Left multiplication by a local output unitary only conjugates the retained output
    and leaves $Z_{\cal{S}}$ unchanged. Minimizing over these local unitaries gives the
    same Lipschitz bound with $d_\otimes$.
    Coupling $G=f_h(H)$ with its Haar preimage, as
    in~\cite[Lemma~6.8]{fefferman2024anti}, now yields
    \begin{equation} \label{eq:rounded_transfer_tail}
        \Pr_{G\sim\mu_h}\bigl(Z_{\cal{S}}(G)\le x\bigr)
        \le A_0(x+144h)^\alpha.
    \end{equation}
    In particular, with $A=A_0\,145^\alpha$, the first bound
    in~\Cref{eq:discrete_transfer_bounds} holds, and the tail is at most $Ax^\alpha$
    whenever $x\ge h$.

    Choose $0<s<\alpha$ and write $W=\max\{h,Z_{\cal{S}}(G)\}$.
    Integrating the lower tail gives
    \begin{equation}
        \mathbb{E}W^{-s}
        \le1+s\int_0^1 x^{-s-1}\Pr(W<x)\,dx
        \le1+As\int_h^1 x^{\alpha-s-1}\,dx
        \le1+\frac{As}{\alpha-s}.
    \end{equation}
    This proves the second bound with $K=1+As/(\alpha-s)$.
\end{proof}

\subsection{Propagation through retained circuits}

We now establish signal propagation along a path, following \cite[Sections~5 and~6.3.2]{fefferman2024anti}. 
We then combine this bound
with~\Cref{thm:learning_forward_general} to obtain exact proper learning for architectures
with good forward lightcone structure.

The additive scale in~\Cref{eq:rounded_transfer_tail} prevents an unmodified Haar
anti-concentration bound at arbitrarily small $x$. In the proof below, we treat a path
step with determinant at most $h$ as a failure and use the clipped negative-moment
bound to control the cumulative contraction. This avoids conditioning the gate
distribution on all path determinants exceeding $h$.

\begin{lemma}[Discrete signal propagation] \label{lem:discrete_signal_propagation}
    Let $R$ have a fixed two-qubit gate architecture with independent gates drawn from
    $\mu_h$. Fix an input qubit $a$, an output qubit $z$ reachable from $a$, and a state
    $\tau$ on the other inputs, all before sampling the gates.
    Choose an architecture-determined path from $a$ to $z$ traversing $d\ge1$ gates.
    There are absolute constants $B>1$ and $r>0$ such that, for every $\gamma\in(0,1)$,
    the induced channel $\cal{E}_{R,a\to z}^\tau$ is
    $B^{-d}\gamma^r$-invertible in Frobenius norm, except with probability at most
    \begin{equation} \label{eq:discrete_propagation_failure}
        \gamma+Adh^\alpha.
    \end{equation}
\end{lemma}
\begin{proof}
    Reveal the circuit gates in a fixed order consistent with the circuit's causal structure,
    and label the gates on the chosen path $G_1,\ldots,G_d$. Before revealing $G_i$, the channel
    from the varying input to its two inputs is fixed. Until a preceding path step has failed,
    the corresponding subspace $\cal{S}_i$ has dimension three.
    Let $Z_i=Z_{\cal{S}_i}(G_i)$ and stop the argument at the first event $Z_i\le h$.
    Conditional on any history before stopping, \Cref{lem:discrete_transfer}
    applies because $G_i$ is independent of that history.
    Thus the probability of stopping is at most $Adh^\alpha$.

    For completeness, the deterministic norm estimate is the one
    in~\cite[Proposition~5.2 and proof of Theorem~1.1]{fefferman2024anti}.
    The singular values of $M_{G_i}$ are at most $2$, so its smallest singular
    value is at least $Z_i/4$. The Frobenius norm of a two-qubit input difference is
    at least half the norm of its marginal on the incoming path wire.
    Consequently, if $v_{i-1}$ and $v_i$ are the incoming and outgoing single-qubit
    state differences, then, simultaneously for all pairs of input states,
    \begin{equation} \label{eq:discrete_path_step}
        \norm{v_i}_F\ge\frac{Z_i}{8}\norm{v_{i-1}}_F.
    \end{equation}
    Initially the input differences have dimension three. A step with $Z_i>h$ preserves
    injectivity, which justifies the dimension assertion at every subsequent step.

    Put $W_i=\max\{h,Z_i\}$ at each step reached before stopping, including the stopping
    step, and put $W_i=1$ after stopping. By~\Cref{lem:discrete_transfer}, the conditional
    negative moments are at most $K$. Iterated conditional expectations, without any
    independence assumption on the $W_i$, give
    \begin{equation}
        \mathbb{E}\prod_{i=1}^d W_i^{-s}\le K^d,\qquad
        \Pr\left(\prod_{i=1}^d W_i<(\gamma/K^d)^{1/s}\right)\le\gamma.
    \end{equation}
    Outside this event and the stopping event, $W_i=Z_i$ throughout the path.
    Multiplying~\Cref{eq:discrete_path_step} proves the claim with
    $B=8K^{1/s}$ and $r=1/s$.
\end{proof}

\begin{proposition}[Learning a fixed architecture] \label{prop:discrete_architecture_learning}
    Fix a two-qubit gate architecture on $n$ qubits with $m\ge1$ gates, maximum depth $D$,
    and good forward lightcone structure. Let $\delta,\eta\in(0,1)$ and choose
    \begin{equation} \label{eq:discretization_resolution}
        h=\min\left\{\frac12,\left(\frac{\delta}{2AmD}\right)^{1/\alpha}\right\}.
    \end{equation}
    Draw each gate independently from $\mu_h$.
    Given the architecture, query access to its unitary channel, the finite gate list
    $\cal{G}_h$, and exact implementations of its inverses as
    in~\Cref{thm:learning_forward_general}, exact proper learning succeeds with probability
    at least $1-\delta-\eta$ over the gates and measurements.
    If $D=O(\log n)$ and $\delta^{-1}$ is polynomial in $n$, the query complexity is
    polynomial in $n$ and $\log(2/\eta)$.
\end{proposition}
\begin{proof}
    Fix a complete front-pivot order and witnesses using only the architecture.
    There are $m$ retained channels required by~\Cref{thm:learning_forward_general},
    all with the other inputs in $\ket{0}$. Choose a shortest path for each channel;
    each has length at most $D$.
    Apply~\Cref{lem:discrete_signal_propagation} to each nonzero-length path with
    $\gamma=\delta/(2m)$. A union bound and~\Cref{eq:discretization_resolution} give
    simultaneous propagation except with probability at most
    \begin{equation}
        \frac{\delta}{2}+AmDh^\alpha\le\delta.
    \end{equation}
    The retained circuits share gate labels, but no independence between these events
    is required. Their marginal gate distributions are products of $\mu_h$ because the
    removal order is architecture-determined.
    On this event the nonincreasing propagation function can be taken as
    \begin{equation}
        L(0)=1,\qquad L(d)=B^{-d}\left(\frac{\delta}{2m}\right)^r
        \quad\text{for integers }d\ge1.
    \end{equation}
    The gate list is $h$-spaced in $d_\otimes$, so
    \Cref{thm:learning_forward_general} applies with spacing $\epsilon=h$ and
    measurement failure probability $\eta$.
    Since $m=O(nD)$, the chosen $h$ and $L(D)$ are inverse polynomial in the stated
    regime, and~\Cref{eq:discrete_gate_count} makes the query bound polynomial.
\end{proof}

\noindent The proposition imposes no geometric or layering assumptions.
The polynomial bound also covers the number of gate trials and the lightcone computations.
As in~\Cref{thm:learning_forward_general}, the gate list and exact known-gate implementations
are supplied; we do not assert an efficient uniform construction or synthesis of the
discretizations. 

We conclude with the quantitative version of~\Cref{cor:single_overlap_learning}.
\begin{corollary}[Proper learning below the threshold] \label{cor:single_overlap_learning_formal}
    Fix $b>0$ and let $n\to\infty$ through even integers. Sample a depth-$d$ layout from the
    two-local single-overlap ensemble, with $1\le d\le t(n,2)-\Delta_2$ as
    in~\Cref{thm:lightcone_threshold}. Draw its gates independently from the well-spaced
    Haar discretization of~\Cref{def:separated_haar_discretization}, using the resolution
    in~\Cref{eq:discretization_resolution} with $m=nd/2$, $D=d$, and $\delta=n^{-b}$.
    Given the layout, query access to the circuit, the finite gate list, and exact
    implementations of its inverses as in~\Cref{thm:learning_forward_general}, exact proper
    learning is possible with query complexity polynomial in $n$ and $\log(2/\eta)$
    for every $\eta\in(0,1)$, with success probability at least
    $1-o(1)-n^{-b}-\eta$ over the layout, gates, and measurements.
\end{corollary}
\begin{proof}
    By~\Cref{thm:lightcone_threshold}, with probability $1-o(1)$ over the layout, the two
    forward lightcones immediately after every gate are distinct. Such a layout has good
    forward lightcone structure, so its complete front-pivot order and witnesses can be
    fixed independently of the gate labels. Conditional on any such layout,
    \Cref{prop:discrete_architecture_learning} applies with $m=nd/2$, $D=d$, and
    $\delta=n^{-b}$. Since $d=O(\log n)$, the query bound is polynomial in the stated
    parameters. Combining the architecture, gate-sampling, and measurement failure
    probabilities proves the claim.
\end{proof}
\noindent The single-overlap condition is used only to establish good forward lightcone
structure with high probability through~\Cref{thm:lightcone_threshold}.

\end{appendices}

\end{document}